\documentclass[showpacs,aps,floatfix,onecolumn,reprint]{revtex4-1}%
\usepackage{amssymb}
\usepackage{amsfonts}
\usepackage{amsmath}
\usepackage{graphicx}%
\usepackage{epstopdf}
\providecommand{\U}[1]{\protect\rule{.1in}{.1in}}
\providecommand{\U}[1]{\protect\rule{.1in}{.1in}}
\newtheorem{theorem}{Theorem}

\newtheorem{claim}[theorem]{Claim}

\newtheorem{lemma}[theorem]{Lemma}

\newenvironment{proof}[1][Proof]{\noindent\textbf{#1.} }{\ \rule{0.5em}{0.5em}}
\begin{document}
\preprint{cond-mat}
\title[Short title for running header]{Cusp singularities in boundary-driven diffusive systems}
\author{Guy Bunin}
\affiliation{Technion -- Israel Institute of Technology, Haifa 32000, Israel}
\author{Yariv Kafri}
\affiliation{Technion -- Israel Institute of Technology, Haifa 32000, Israel}
\author{Daniel Podolsky}
\affiliation{Technion -- Israel Institute of Technology, Haifa 32000, Israel}

\begin{abstract}
Boundary driven diffusive systems describe a broad range of transport
phenomena. We study large deviations of the density profile in these systems,
using numerical and analytical methods. We find that the large deviation may
be non-differentiable, a phenomenon that is unique to non-equilibrium systems,
and discuss the types of models which display such singularities.\ The
structure of these singularities is found to generically be a cusp, which can
be described by a Landau free energy or, equivalently, by catastrophe theory.
Connections with analogous results in systems with finite-dimensional phase
spaces are drawn.

\end{abstract}
\volumeyear{year}
\volumenumber{number}
\issuenumber{number}
\eid{identifier}
\startpage{1}
\endpage{4}
\maketitle

\section{Introduction and framework}

The dynamics in many systems of physical interest are described by a field
$\rho(x,t)$, with large-scale conserving diffusive behavior and noise. For
example, $\rho(x,t)$ could describe the density of diffusing particles, the
local temperature in a heat transport experiment, or any other field which
behaves diffusively. For such systems, when the interactions are short range,
it is accepted
\cite{Derrida_review,spohn_book,JSP_quantum,BertiniPRL,TKL_long}, that the
large-scale behavior of the current obeys Fick's- (or Ohm's- or Fourier's-)
law with noise. Here our interest is in transport experiments, where the
system is attached to reservoirs, whose effect is to fix the value of $\rho$
at the boundaries, resulting in a net current flowing down the gradient.

For such systems the density $\rho\left(  x,t\right)  $ and the current
$J\left(  x,t\right)  $ satisfy the conservation relation%
\begin{equation}
\partial_{t}\rho+\partial_{x}J=0\ , \label{eq:conservation}%
\end{equation}
where%
\begin{equation}
J=\mathbf{-}D\left(  \rho\left(  x,t\right)  \right)  \partial_{x}\rho\left(
x,t\right)  +\sqrt{\sigma\left(  \rho\left(  x,t\right)  \right)  }\eta\left(
x,t\right)  \ . \label{eq:J_def}%
\end{equation}
Here $D\left(  \rho\left(  x,t\right)  \right)  $ is a density-dependent
diffusivity function, and $\sigma\left(  \rho\left(  x,t\right)  \right)  $
controls the amplitude of the white noise $\eta\left(  x,t\right)  $, which
satisfies $\left\langle \eta\left(  x,t\right)  \right\rangle =0$ and
$\left\langle \eta\left(  x,t\right)  \eta\left(  x^{\prime},t^{\prime
}\right)  \right\rangle =N^{-1}\delta\left(  x-x^{\prime}\right)
\delta\left(  t-t^{\prime}\right)  $. At temperatures well-above any
phase-transition, which we study here, $D\left(  \rho\right)  $ and
$\sigma\left(  \rho\right)  $ are smooth functions, and $D>0$. For simplicity
we consider one dimension, where the distance is rescaled by the system size
$N$, so that $0\leq x\leq1$, and time is rescaled by $N^{2}$. The small
$N^{-1}$ term in the noise is a direct consequence of this coarse-graining.
$D\left(  \rho\right)  $ and $\sigma\left(  \rho\right)  $ are related via a
fluctuation-dissipation relation, which for particle systems reads
$\sigma\left(  \rho\right)  =2k_{B}T\rho^{2}\kappa\left(  \rho\right)
D\left(  \rho\right)  $ where $\kappa\left(  \rho\right)  $ is the
compressibility \cite{Derrida_review}. The system is attached to reservoirs at
the boundaries $x=0,1$, which act as boundary conditions (BCs), $\rho\left(
x=0,t\right)  =\rho_{L}$ and $\rho\left(  x=1,t\right)  =\rho_{R}$.\ If
$\rho_{L}\neq\rho_{R}$ a current is induced through the system, driving it out
of equilibrium. For applications of Eq. (\ref{eq:J_def}) to transport
phenomena, including electronic systems, ionic conductors, and heat
conduction, see for example \cite{JSP_quantum,ionic_conductors,KMP}.

It is natural to ask for the probability of a density profile $\rho_{f}\left(
x\right)  $ in the steady-state. It is known that the probability distribution
assumes the form $P[\rho_{f}]\sim e^{-N\phi\lbrack\rho_{f}]}$, where
$\phi\lbrack\rho_{f}]$ is known as the large deviation functional (LDF), and
the $N$ prefactor is due to the small noise. $P[\rho_{f}]$ is the subject of
this work. As seen, out of equilibrium the LDF $\phi\left[  \rho_{f}\right]  $
plays the role of the free-energy density in equilibrium \cite{Derrida_review}%
. It is by now well established that, in contrast to equilibrium, out of
equilibrium long range correlations build-up
\cite{dorfman_long_range,SSEP_spohn} and the LDF is non-local
\cite{SSEP_spohn,SSEP_large_dev}. Moreover, there is now a general framework
for calculating $\phi\left[  \rho_{f}\right]  $. As detailed below it involves
finding the most probable history leading to $\rho_{f}$. In spite of this, the
properties of $\phi\left[  \rho_{f}\right]  $ remain poorly understood. Much
of what is known is based on a handful of exact solutions for specific models
\cite{SSEP_large_dev,BertiniJstat,KMP_large_dev}, and numerical techniques
\cite{our_numerics}. In particular, it was recently realized that $\phi\left[
\rho_{f}\right]  $ can be non-differentiable \cite{ours_short}.

Here we discuss in detail the occurrence and structure of such singular
behavior in the class of models defined above. We refer to it as a Large
Deviation Singularity (LDS). This is very different from the equilibrium case,
where smooth dynamics (i.e. when $D\left(  \rho\right)  $ and $\sigma\left(
\rho\right)  $ are smooth and $D>0$) lead to a smooth LDF $\phi\left[
\rho_{f}\right]  $.

The general occurrence of non-differentiable LDFs in low-noise Langevin
equations was discovered by Graham and T\'{e}l \cite{Graham_Tel}. LDSs were
consequently widely discussed in the literature
\cite{Noise_non_lin_book,analog_exp,Maier_Stein,Dykman,Dykman_cusp_structure},
demonstrated in experiments \cite{Noise_non_lin_book,analog_exp}, and shown to
affect quantities such as barrier crossing rates \cite{Maier_Stein}. In
addition to these, and more closely related to the present work, an LDS was
proven to exist in the asymmetric exclusion process \cite{ASEP_transition}, a
specific model of diffusing particles, where (unlike in the present
paper)\ the particles are also subject to an driving field in the bulk. This
is perhaps the first known microscopic lattice-gas model for which the
continuum limit was proven to feature an LDS. The proof hinges on the exact
solvability of the model, and it is not clear what other models of that family
will show this behavior. The general conditions for LDSs to occur in fields,
and the structure of the singularities remains largely unknown.

In this paper, we achieve the following.

(1)\ \emph{Existence of non-differentiable LDFs for boundary-driven
systems}\ -- we show that in some boundary-driven diffusive models the LDF is
non-differentiable. This includes an exactly solvable model, where $D=1$ and
$\sigma=\rho^{2}+1$, and a boundary-driven Ising model, with conserving
dynamics in the bulk. The phenomenon is general and robust, and expected to be
found in models where $\sigma\left(  \rho\right)  /D\left(  \rho\right)  $ has
a minimum which is deep enough. The profiles $\rho_{f}\left(  x\right)  $ at
which the derivative $\delta\phi/\delta\rho$ is discontinuous are found to
have a typical shape, as shown in Fig. \ref{fig:qs_cusp}(b) and
\ref{fig:two_paths_bdi}. The jump in the derivative is due to a change in the
form of the most probable history $\rho\left(  x,t\right)  $ leading up to
$\rho_{f}\left(  x\right)  $. It stems from the existence of regions in the
space of $\rho_{f}\left(  x\right)  $ where multiple locally minimizing
histories lead to a single $\rho_{f}\left(  x\right)  $. A short account of
these results was given in \cite{ours_short}.

(2) \emph{Structure of singularities in phase-space} -- we study the singular
structures in phase-space. Two-dimensional cross-sections, as shown in Fig.
\ref{fig:qs_cusp}(c) and \ref{fig:kls_cusp}(a), are illuminating. They show
the regions in the $\rho_{f}\left(  x\right)  $ space with multiple locally
minimizing histories. The boundaries of this region are known as
\emph{caustics}. The points where $\delta\phi/\delta\rho$ jumps occur when the
globally minimizing history changes. These points form the \emph{transition
line}. The transition line and caustics end at a single point, analogous to a
second-order phase transition. We show that the structure near this point is
similar to the one described by a Landau mean-field theory, or by a \emph{cusp
singularity} in catastrophe theory. One outcome of this theory is the
prediction that at the second-order-like point the probability $P[\rho_{f}]$
scales in a non-analytic way with the system size $N$. Specifically, instead
of the expected series $\ln P[\rho_{f}]=-N\phi\lbrack\rho_{f}]+O\left(
N^{0}\right)  $, the series will have an additional logarithmic correction
$\ln P[\rho_{f}]=-N\phi\lbrack\rho_{f}]+1/4\ln N+O\left(  N^{0}\right)  $. The
prefactor $1/4$ is universal, depending only on the symmetries of the systems.

(3)\ \emph{Relation to finite dimensional theory -- }we show that much of the
physics can be understood by introducing simple toy models with as low as two
degrees of freedom.

We stress that the singularities discussed here are different in nature from
those found for global quantities such as the current
\cite{Bertini_current_phase_trans,Bodineau,Merhav_kafri,hurtado_current_KMP}.
In those case the probability of a configuration can be smooth in phase-space,
but the optimizing configuration can change abruptly.

The paper is arranged as follows: In Sec. \ref{sec:GandTel} we present an
example of an LDS in a model with a single degree of freedom, as well as the
background to the general theory. In Sec. \ref{sec:LDS_exist} we demonstrate
the existence of LDSs in models of the family discussed here. We introduce the
two example models which are studied throughout the paper. We show,
analytically for one model and numerically for the other, that LDSs do indeed
exist, and indicate where and under what conditions they are expected. In Sec.
\ref{sec:cusp_structure} we study the structure of the region, and the effect
of this structure on the dependence of the probability $P[\rho_{f}]$ on the
system size. In Sec. \ref{sec:toy_model} we introduce a model with just two
phase-space dimensions which, as we show, captures much of the behavior of the
full infinite-dimensional field model. In Sec. \ref{sec:discussion} we
conclude and discuss future research directions.

\section{LDSs in simple models\label{sec:GandTel}}

Before discussing LDSs in the model defined above, we recall the simplest
example of such a phenomenon, which occurs for a single particle moving on a
ring. As originally discussed by Graham and T\'{e}l \cite{Graham_Tel}, when
such a system is driven out of equilibrium the gradient of the LDF becomes
discontinuous. It is instructive to see how this singularity arises, despite
the fact that many of the features are different from the singularities
discussed in this paper. \ 

Consider a particle moving on a one dimensional ring $x\in\left[  0,1\right]
$, subject to the Langevin equation%
\[
\frac{dx}{dt}=f_{0}-V^{\prime}\left(  x\right)  +\sqrt{\varepsilon}\eta\left(
t\right)  \ ,
\]
where $\left\langle \eta\left(  t\right)  \eta\left(  t^{\prime}\right)
\right\rangle =2\delta\left(  t-t^{\prime}\right)  $, $f_{0}$ is a constant,
$V\left(  x\right)  $ a periodic function on the ring, and $\varepsilon$ is a
small number which plays an analogous role to $N^{-1}$ in Eq. (\ref{eq:J_def}%
). For $f_{0}\neq0$, the total force $F\left(  x\right)  =f_{0}-V^{\prime
}\left(  x\right)  $ is not derivable from a potential. It is useful to define
the integral $U\left(  x\right)  =\int_{0}^{x}F\left(  x^{\prime}\right)
dx^{\prime}$ for $x\in\left[  0,1\right]  $ which is no longer periodic in
$x$. Consider the case with $U\left(  x\right)  $ shown in Fig.
\ref{fig:graham_tel_ring}.%

%TCIMACRO{\FRAME{ftbpFU}{1.7113in}{1.0488in}{0pt}{\Qcb{Simple model with
%singular LDF. The gray curve is $U\left(  x\right)  =-f_{0}x-V\left(
%x\right)  $, and black curve is the LDF $\phi\left(  x\right)  $. The dashed
%arrows show the most probable path for a particle from A, the local minimum of
%the potential, to a point between B and C. $\phi^{\prime}\left(  x\right)  $
%is discontinuous at point C.}}{\Qlb{fig:graham_tel_ring}}{graham_tel_ring.eps}%
%{\special{ language "Scientific Word";  type "GRAPHIC";
%maintain-aspect-ratio TRUE;  display "USEDEF";  valid_file "F";
%width 1.7113in;  height 1.0488in;  depth 0pt;  original-width 1.3878in;
%original-height 0.8401in;  cropleft "0";  croptop "1";  cropright "1";
%cropbottom "0";  filename 'Graham_Tel_ring.eps';file-properties "XNPEU";}} }%
%BeginExpansion
\begin{figure}
[ptb]
\begin{center}
\includegraphics[
height=1.0488in,
width=1.7113in
]%
{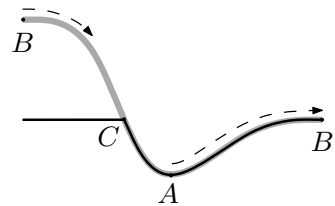}%
\caption{Simple model with singular LDF. The gray curve is $U\left(  x\right)
=-f_{0}x-V\left(  x\right)  $, and black curve is the LDF $\phi\left(
x\right)  $. The dashed arrows show the most probable path for a particle from
A, the local minimum of the potential, to a point between B and C.
$\phi^{\prime}\left(  x\right)  $ is discontinuous at point C.}%
\label{fig:graham_tel_ring}%
\end{center}
\end{figure}
%EndExpansion

As the noise is small (because of the $\varepsilon$ prefactor), the system
spends most of its time near point A. We now want to evaluate the probability
$P\left(  x\right)  \sim e^{-\varepsilon^{-1}\phi\left(  x\right)  }$\ of
reaching any other point, in the low noise limit to leading order in
$\varepsilon^{-1}$. The point B is represented both by $x=0$ and $x=1$.
However, due to the bias force it is easier\ to reach B by moving to the
right. Therefore the probability of reaching point B is given by the Arrenius
factor
\[
P\sim e^{-\varepsilon^{-1}\left[  U\left(  1\right)  -U\left(  A\right)
\right]  }=e^{-\varepsilon^{-1}\phi\left(  B\right)  }\ .
\]
To see why $\phi\left(  x\right)  $ is singular, note that once the particle
has reached this point it may \textquotedblleft roll-down\textquotedblright%
\ to reach all points. Therefore the probability of being between B and C is
equal to this order. Only below C is it preferable to move from point A to the
left, and the probability changes again. Thus for $f_{0}\neq0$ one obtains a
plateau and a discontinuity in $\phi^{\prime}\left(  x\right)  $ at point C.

While the plateau is a rather specific feature of this and similar examples,
the existence of a discontinuity in the derivative $\phi^{\prime}\left(
x\right)  $ is a common feature of non-equilibrium low-noise systems. It
results from a competition between different trajectories which lead to the
same final point. Here, these are trajectories moving to the right and left.

In the context of boundary-driven diffusive systems, the LDFs of all models
which were previously studied exhibited a smooth LDF. This raises the question
of whether, and for what models of this family, will LDSs exist. In addition,
it is interesting to understand the structure of these singularities, and
whether they are similar to what is known for models in a finite-dimensional
space, where the structure can be understood using mean-field, or catastrophe theory.

\subsection{Background theory\label{sec:background}}

We now outline the theoretical tools used below. The average, or most probable
density profile for the system, $\bar{\rho}$, is obtained by solving
$\partial_{x}\left[  D\left(  \bar{\rho}\right)  \partial_{x}\bar{\rho
}\right]  =0$, with $\bar{\rho}\left(  0\right)  =\rho_{L}$ and $\bar{\rho
}\left(  1\right)  =\rho_{R}$ at the boundaries. As the noise is small, the
system spends most of its time close to $\bar{\rho}$. In order to find the
probability of any profile $\rho_{f}\left(  x\right)  $, one must therefore
compute the probability of reaching $\rho_{f}$, starting from $\bar{\rho}$ in
the distant past.

The probability density of a history $\left\{  \rho\left(  x,t\right)
,J\left(  x,t\right)  \right\}  $\ during time $-\infty\leq t\leq0$ is $P\sim
e^{-NS}$ where the action $S$ is given by
\cite{BertiniPRL,BertiniJstat,JSP_quantum,TKL_long,Freidlin_Wentzell}%
\begin{equation}
S=\int_{-\infty}^{0}dt\int_{0}^{1}dx\frac{\left[  J\left(  x,t\right)
+D\left(  \rho\left(  x,t\right)  \right)  \partial_{x}\rho\left(  x,t\right)
\right]  ^{2}}{2\sigma\left(  \rho\left(  x,t\right)  \right)  }%
\ .\label{eq:action}%
\end{equation}
The probability density $P\left[  \rho_{f}\right]  $ of reaching $\rho_{f}$ is
then given by the path integral
\[
P\left[  \rho_{f}\right]  =\int D\rho\int DJe^{-NS\left[  \rho,J\right]  }\ ,
\]
taken over histories satisfying $\partial_{t}\rho+\partial_{x}J=0$, with
initial and final conditions $\rho\left(  x,t\rightarrow-\infty\right)
=\bar{\rho}\left(  x\right)  $, $\rho\left(  x,t=0\right)  =\rho_{f}\left(
x\right)  $, and boundary conditions $\rho\left(  x=0,t\right)  =\rho_{L}$ and
$\rho\left(  x=1,t\right)  =\rho_{R}$. For large $N$ a saddle-point
approximation gives $P\sim e^{-N\phi\left[  \rho_{f}\right]  }$ with
$\phi\left[  \rho_{f}\right]  =\inf_{\rho,J}S$, where the infimum is over all
allowed histories.

In equilibrium (i.e. when $\rho_{L}$ $=\rho_{R}$), the steady-state
probability of a density profile $\rho_{f}\left(  x\right)  $ is easy to
obtain -- the LDF $\phi\left[  \rho_{f}\right]  $ is then given by the
free-energy which is local in $\rho_{f}$, $\phi\lbrack\rho_{f}]=\int f\left(
\rho_{f}\left(  x\right)  ,\bar{\rho}\right)  dx$, where
\begin{equation}
f\left(  \rho,r\right)  \equiv\int_{r}^{\rho}d\rho_{1}\int_{r}^{\rho_{1}}%
d\rho_{2}\frac{2D\left(  \rho_{2}\right)  }{\sigma\left(  \rho_{2}\right)
}\ . \label{eq:free_energy_density}%
\end{equation}
Note that in this case $\bar{\rho}$ is constant, $\bar{\rho}=\rho_{L}$
$=\rho_{R}$. By contrast, the steady-state probability distribution away from
equilibrium is notoriously hard to compute, and \emph{very} different from the
naive guess $\phi\lbrack\rho_{f}]=\int f\left(  \rho_{f}\left(  x\right)
,\bar{\rho}\left(  x\right)  \right)  dx$, now with space dependent $\bar
{\rho}\left(  x\right)  $. In fact, as stated above, $\phi\lbrack\rho_{f}]$ is non-local.

\section{Existence of LDS\label{sec:LDS_exist}}

As is clear from Eq. (\ref{eq:free_energy_density}), LDSs\ cannot exist in
equilibrium if $D\left(  \rho\right)  $ and $\sigma\left(  \rho\right)  $ are
smooth, positive functions. We now turn to discuss non-equilibrium cases where
they can exist. In previously studied exactly solvable non-equilibrium models
\cite{SSEP_large_dev,KMP_large_dev}, the action $S$\ in Eq. (\ref{eq:action}%
)\ has a single local minimal history leading to $\rho_{f}$, and $\phi\left[
\rho_{f}\right]  $ is then a smooth functional. However, this need not always
be the case, and there can be multiple local minima to $S$ to the same
$\rho_{f}$. For example, in the model of a single particle described in Sec.
\ref{sec:GandTel}, these are trajectories corresponding to the particle moving
left or right on the ring. When this occurs, the global minimum can switch
between the different local minima (as it does at point C in Fig.
\ref{fig:graham_tel_ring}). For fields, this is accompanied by a jump in the
functional derivative of the large-deviation $\delta\phi/\delta\rho_{f}$. This
is reminiscent of the mechanism for a first-order phase transition in equilibrium.

It is unclear which models display an LDS. It is known that models which
feature more than one fixed point of the zero-noise dynamics generically
display LDSs \cite{Graham_Tel}. In the models we study here, however,
$\bar{\rho}$ is the only zero-noise fixed-point. Therefore, the absence of
LDSs in previously studies models of this class is not surprising
\cite{SSEP_large_dev,KMP_large_dev}.

In this section we discuss the existence of LDSs in two models. One of the
models originates by taking the continuum limit of a \textquotedblleft
microscopic\textquotedblright\ lattice gas model, the driven Ising model. The
other has the advantage of being exactly solvable. As explained in Sec.
\ref{sec:background}, a model is defined by two functions, the diffusivity
$D\left(  \rho\right)  $ and the noise strength $\sigma\left(  \rho\right)  $.
These are shown in Fig. \ref{fig:models_sig_d}(a) for a specific set of
parameters of the driven Ising model, and \ref{fig:models_sig_d}(b) for the
analytically solvable model, referred to as the quadratic-$\sigma$
(QS)\ model. We first define the models, and then discuss their properties.

\subsection{Definition of models}

Here we define two models for which we demonstrate the existence of LDSs. As a
common feature, both models display a pronounced dip in the function
$\sigma\left(  \rho\right)  /D\left(  \rho\right)  $.%
%TCIMACRO{\FRAME{ftbpFU}{3.3167in}{1.3589in}{0pt}{\Qcb{Model definitions. The
%functions $\sigma\left(  \rho\right)  $ and $D\left(  \rho\right)  $ for (a)
%the BDI model, and (b) the QS model.}}{\Qlb{fig:models_sig_d}}%
%{models_sig_d.eps}{\special{ language "Scientific Word";  type "GRAPHIC";
%maintain-aspect-ratio TRUE;  display "USEDEF";  valid_file "F";
%width 3.3167in;  height 1.3589in;  depth 0pt;  original-width 6.9655in;
%original-height 2.8374in;  cropleft "0";  croptop "1";  cropright "1";
%cropbottom "0";  filename 'models_sig_D.eps';file-properties "XNPEU";}} }%
%BeginExpansion
\begin{figure}
[ptb]
\begin{center}
\includegraphics[
height=1.3589in,
width=3.3167in
]%
{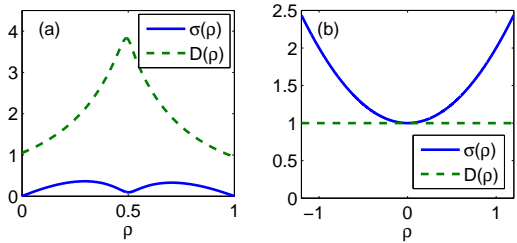}%
\caption{Model definitions. The functions $\sigma\left(  \rho\right)  $ and
$D\left(  \rho\right)  $ for (a) the BDI model, and (b) the QS model.}%
\label{fig:models_sig_d}%
\end{center}
\end{figure}
%EndExpansion

\subsubsection{The boundary-driven Ising (BDI) model}

This is a lattice gas with on-site exclusion and nearest-neighbor interaction.
It corresponds to the\ Katz-Lebowitz-Spohn model \cite{KLS} with zero bulk
bias. The model is defined on a 1d lattice with sites $i=1,..,N$, each of
which can be either occupied (\textquotedblleft1\textquotedblright) or empty
(\textquotedblleft0\textquotedblright). The model depends on two rate
parameters $\delta$ and $\varepsilon$. The jump rate from site $i$ to site
$i+1$ depends on the occupation at sites $i-1$ to $i+2$ according to the
following rules:
\begin{align*}
&  0100\overset{1+\delta}{\rightarrow}0010,\ 1101\overset{1-\delta
}{\rightarrow}1011\ ,\\
&  1100\overset{1+\varepsilon}{\rightarrow}1010,\ 1010\overset{1-\varepsilon
}{\rightarrow}1100\ ,
\end{align*}
and their spatially inverted counterparts with identical rates.

For equilibrium BCs, e.g., periodic BCs, the dynamics admits an Ising
probability distribution $P\propto\exp\left(  -\beta E\right)  $ with
\[
E=\sum_{i}\left(  1-2n_{i}\right)  \left(  1-2n_{i+1}\right)  +\mu\sum
_{i}\left(  1-2n_{i}\right)  \ .
\]
This energy describes nearest neighbor interactions, and a chemical potential
term. $\beta$ is related to $\varepsilon$ by $\exp\left(  4\beta\right)
=\left(  1+\varepsilon\right)  /\left(  1-\varepsilon\right)  $, and $\mu$
fixes the average density. The parameter $\delta$ does not affect the
stationary state, but does enter into the dynamical behavior of the model. For
each parameter set $\left(  \varepsilon,\delta\right)  $ one can write
implicit analytic equations for $D\left(  \rho\right)  ,\sigma\left(
\rho\right)  $ which can then be inverted numerically. The calculation is
described in Appendix \ref{sec:appendix_KLS_D_sig}.\ Fig.
\ref{fig:models_sig_d}(a) shows $D\left(  \rho\right)  $ and $\sigma\left(
\rho\right)  $ for $\left(  \varepsilon,\delta\right)  =\left(
0.05,0.995\right)  $. As can be seen, $D\left(  \rho\right)  $ is peaked and
$\sigma\left(  \rho\right)  $ has a local minimum around $\rho=1/2$. This will
be a key feature of the model.

\subsubsection{\emph{The} \emph{quadratric-}$\sigma$ (QS) model}

The model is defined by constant $D$ and $\sigma\left(  \rho\right)  =c\left(
\rho-b\right)  ^{2}+a$, with $a,c>0$, so that $\sigma\left(  \rho\right)  $ is
a parabola clear above the axis. Upon the rescaling
\begin{align*}
\rho &  \rightarrow\sqrt{\frac{a}{c}}\rho+b\ ,\ J\rightarrow D\sqrt{\frac
{a}{c}}J~,\ \\
x &  \rightarrow x\ ,\ t\rightarrow t/D\ ,\ S\rightarrow cS/D
\end{align*}
the model can be brought to a standard form defined by $D=1$ and
$\sigma\left(  \rho\right)  =\rho^{2}+1$, see Fig. \ref{fig:models_sig_d}(b).
This standard form will be used throughout the text. Note that the BCs of the
density map accordingly.

The QS model has the advantage that it is analytically tractable
\cite{KMP_large_dev}: the LDF is given by $\phi\left[  \rho_{f}\right]
=\min\phi_{ext}$, where $\phi_{ext}$ are extremal values of the action given
by%
\begin{equation}
\phi_{ext}=\int_{0}^{1}dx\left\{  f\left(  \rho_{f}\left(  x\right)  ,g\left(
x\right)  \right)  -\ln\frac{g^{\prime}\left(  x\right)  }{\bar{\rho}^{\prime
}\left(  x\right)  }\right\}  \ .\label{eq:quad_sig_LDF_expression}%
\end{equation}
Here $f\left(  \rho,g\right)  $ is defined in Eq.
(\ref{eq:free_energy_density}) and $g\left(  x\right)  $ is an auxiliary
function satisfying the differential equation%
\begin{equation}
0=\frac{g\left(  x\right)  -\rho_{f}\left(  x\right)  }{\sigma\left(  g\left(
x\right)  \right)  }-\frac{g^{\prime\prime}\left(  x\right)  }{\left[
g^{\prime}\left(  x\right)  \right]  ^{2}}\ ,\label{eq:general_ODE}%
\end{equation}
with BCs $g\left(  0\right)  =\rho_{L}$, and $g\left(  1\right)  =\rho_{R}%
$.\ Note that as $D=1$, the most probable configuration $\bar{\rho}\left(
x\right)  $ is linear, with $\bar{\rho}\left(  0\right)  =\rho_{L}$ and
$\bar{\rho}\left(  1\right)  =\rho_{R}$. Each of the solutions of Eq.
(\ref{eq:general_ODE}), when used in Eq. (\ref{eq:quad_sig_LDF_expression}),
gives $\phi_{ext}$ of an extremal path \cite{foot_extremize_S}.

\subsection{The use of cross-sections}

Below we demonstrate the existence of LDSs in the models defined above. As the
phase-space is infinite dimensional, the structure of $\phi$ is hard to
visualize. For many purposes it is sufficient to consider two-dimensional
cross-sections of the infinite-dimensional phase-space.

To this end, in most of what follows, out of the phase-space of final profiles
$\rho_{f}\left(  x\right)  $ we restrict ourselves to those parametrized by
just two variables, of the form
\begin{equation}
\rho_{f}\left(  x\right)  =\bar{\rho}\left(  x\right)  +\alpha_{1}\sin\pi
x+\alpha_{2}\sin2\pi x\ .\label{eq:sin_interp}%
\end{equation}
This is a cross-section in the phase-space of final states. (Note that in
Appendix \ref{sec:appendix_quad_sig} we use a different form.) It will be more
convenient to parametrize these profiles using $\rho_{f}\left(  1/3\right)  $
and $\rho_{f}\left(  2/3\right)  $ instead of $\alpha_{1},\alpha_{2}$. We
stress that this choice is rather arbitrary and that the singularity described
occupies a space of co-dimension 1 in the infinite-dimensional phase space.

In order to visualize trajectories $\rho\left(  x,t\right)  $ leading to
$\rho_{f}\left(  x\right)  $, we plot $\rho\left(  x=2/3,t\right)  $ against
$\rho\left(  x=1/3,t\right)  $. Note that here we do not constrain
$\rho\left(  x,t\right)  $ at intermediate times to be of the form in Eq.
(\ref{eq:sin_interp}).

\subsection{Non-unique path minimizers and the LDS}

As we now show, in both the BDI and the QS models, there are certain states
$\rho_{f}$ for which there exists more than a single history $\rho\left(
x,t\right)  $ that extremalizes the action in Eq. (\ref{eq:action}). In order
to find multiple extremal solutions we use different techniques, depending on
the model.

In the QS model we look for solutions to the differential equation
(\ref{eq:general_ODE}). These are found using a shooting method
\cite{numerical_recepies}, in which Eq. (\ref{eq:general_ODE}) is integrated
from $x=0$ to $x=1$, with initial conditions $g\left(  0\right)  =\rho_{L}$,
and $g^{\prime}\left(  0\right)  =c$. The values of $c$ are scanned
systematically to find all solutions where $g\left(  1\right)  =\rho_{R}$. In
this way all solutions of Eq. (\ref{eq:general_ODE}) are obtained.%
%TCIMACRO{\FRAME{ftbpFU}{3.7182in}{1.8968in}{0pt}{\Qcb{QS cusp. (a) A profile
%$\rho_{f}$ (solid line) for which Eq. (\ref{eq:general_ODE}) has a single
%solution (dashed line). (b) A profile $\rho_{f}$ for which Eq.
%(\ref{eq:general_ODE}) has three solutions. $\rho_{f}$ of panel (a) is shown
%for comparison (dotted line). (c) A cross-section - the density $\rho
%_{f}\left(  x=2/3\right)  $ vs. $\rho_{f}\left(  x=2/3\right)  $, for
%configurations of the form given by Eq. (\ref{eq:sin_interp}). The region of
%multiple solutions (gray), and the switching line (dashed line). Crosses
%denote the locations in this cross section of the profiles shown in panels
%(a,b).}}{\Qlb{fig:qs_cusp}}{qs_cusp.eps}%
%{\special{ language "Scientific Word";  type "GRAPHIC";
%maintain-aspect-ratio TRUE;  display "USEDEF";  valid_file "F";
%width 3.7182in;  height 1.8968in;  depth 0pt;  original-width 6.9655in;
%original-height 2.8374in;  cropleft "0";  croptop "1";  cropright "1";
%cropbottom "0";  filename 'QS_cusp.eps';file-properties "XNPEU";}} }%
%BeginExpansion
\begin{figure}
[ptb]
\begin{center}
\includegraphics[
height=1.8968in,
width=3.7182in
]%
{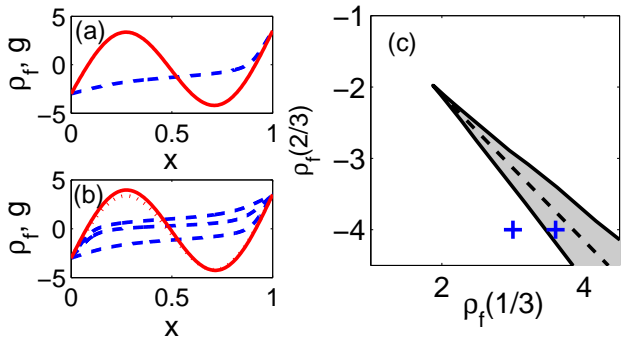}%
\caption{QS cusp. (a) A profile $\rho_{f}$ (solid line) for which Eq.
(\ref{eq:general_ODE}) has a single solution (dashed line). (b) A profile
$\rho_{f}$ for which Eq. (\ref{eq:general_ODE}) has three solutions. $\rho
_{f}$ of panel (a) is shown for comparison (dotted line). (c) A cross-section
- the density $\rho_{f}\left(  x=2/3\right)  $ vs. $\rho_{f}\left(
x=2/3\right)  $, for configurations of the form given by Eq.
(\ref{eq:sin_interp}). The region of multiple solutions (gray), and the
switching line (dashed line). Crosses denote the locations in this cross
section of the profiles shown in panels (a,b).}%
\label{fig:qs_cusp}%
\end{center}
\end{figure}
%EndExpansion

Using final profiles $\rho_{f}$ of the form of Eq. (\ref{eq:sin_interp}) with
$\rho_{L}=-3,\rho_{R}=3$, we find two distinct behaviors. For final profiles
which lie in the white region of Fig. \ref{fig:qs_cusp}(c) we obtain a single
solution to Eq. (\ref{eq:general_ODE}), as illustrated in Fig.
\ref{fig:qs_cusp}(a). In contrast, for final profiles in the gray region Fig.
\ref{fig:qs_cusp}(c) we find three solutions, see Fig. \ref{fig:qs_cusp}(b).
Of these three solutions, two correspond to local minima of the action Eq.
(\ref{eq:action}) and one to a saddle-point. Among the two minima, one is
lower than the other except along the \emph{switching line}, marked by a
dashed line in Fig. \ref{fig:qs_cusp}(c), where they are equal. On this line
the global minimum switches from one local minimum to the other. This leads to
a jump in the gradient of the LDF $\delta\phi/\delta\rho_{f}$ across the line.

The phase diagram, shown in Fig. \ref{fig:qs_cusp}(c), is reminiscent of that
obtained from a Landau free-energy. In this analogy, the gray region
corresponds to the free-energy having two local minima, one metastable. The
boundaries of the gray region are then the spinodal lines (where the
metastable minimum disappears), and the switching line corresponds to a
first-order transition (where the metastable and stable minima exchange
roles). The switching line terminates at a point analogous to a critical
point. We examine this issue in detail below, and show that a universal
behavior emerges.

It is natural to ask which BCs admit profiles $\rho_{f}\left(  x\right)  $
with multiple minimizing solutions. In the case of the QS model, we can in
fact show that: \emph{For any BCs }$\rho_{L}\neq\rho_{R}$\emph{, there exists
a profile }$\rho_{f}\left(  x\right)  $\emph{ for which Eq.
(\ref{eq:general_ODE}) is satisfied by more than one solution. }The proof is
given in Appendix \ref{sec:appendix_quad_sig}. This is interesting since it
implies that in this model \emph{even the smallest deviation} of the BCs from
equilibrium leads to the existence of LDSs. The closer the BCs are to
equilibrium, the further the states $\rho_{f}$ are from $\bar{\rho}$ before
multiple solutions exist.%
%TCIMACRO{\FRAME{ftbpFU}{2.7702in}{1.7935in}{0pt}{\Qcb{LDS in the BDI model.
%Two locally minimizing histories leading to the same $\rho_{f}$, plotted at
%different times.}}{\Qlb{fig:two_paths_bdi}}{two_paths_bdi.eps}%
%{\special{ language "Scientific Word";  type "GRAPHIC";
%maintain-aspect-ratio TRUE;  display "USEDEF";  valid_file "F";
%width 2.7702in;  height 1.7935in;  depth 0pt;  original-width 4.5981in;
%original-height 2.8428in;  cropleft "0";  croptop "1";  cropright "1";
%cropbottom "0";  filename 'two_paths_BDI.eps';file-properties "XNPEU";}} }%
%BeginExpansion
\begin{figure}
[ptb]
\begin{center}
\includegraphics[
height=1.7935in,
width=2.7702in
]%
{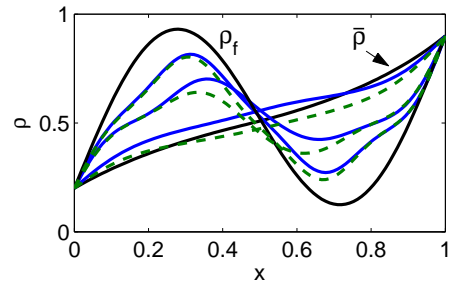}%
\caption{LDS in the BDI model. Two locally minimizing histories leading to the
same $\rho_{f}$, plotted at different times.}%
\label{fig:two_paths_bdi}%
\end{center}
\end{figure}
%EndExpansion
%TCIMACRO{\FRAME{ftbpFU}{3.9241in}{1.9976in}{0pt}{\Qcb{LDS in the BDI model.
%(a) The time evolution of $\rho\left(  1/3,t\right)  $ vs. $\rho\left(
%2/3,t\right)  $ is plotted for the histories of Fig. \ref{fig:two_paths_bdi}.
%(b)\ The switching line\ (dashed line), together with lines of equal $\phi$
%(solid lines).}}{\Qlb{fig:kls_cusp}}{kls_cusp.eps}%
%{\special{ language "Scientific Word";  type "GRAPHIC";
%maintain-aspect-ratio TRUE;  display "USEDEF";  valid_file "F";
%width 3.9241in;  height 1.9976in;  depth 0pt;  original-width 5.6119in;
%original-height 2.8428in;  cropleft "0";  croptop "1";  cropright "1";
%cropbottom "0";  filename 'KLS_cusp.eps';file-properties "XNPEU";}} }%
%BeginExpansion
\begin{figure}
[ptb]
\begin{center}
\includegraphics[
height=1.9976in,
width=3.9241in
]%
{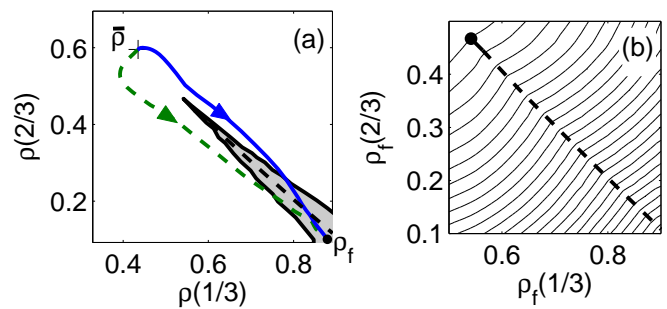}%
\caption{LDS in the BDI model. (a) The time evolution of $\rho\left(
1/3,t\right)  $ vs. $\rho\left(  2/3,t\right)  $ is plotted for the histories
of Fig. \ref{fig:two_paths_bdi}. (b)\ The switching line\ (dashed line),
together with lines of equal $\phi$ (solid lines).}%
\label{fig:kls_cusp}%
\end{center}
\end{figure}
%EndExpansion

We now turn to the BDI model. Here no analytical solution is known, and we
solve for local minimizers of the action $S$, using the numerical technique
described in \cite{our_numerics,ours_short}. Again, for $\rho_{L}\neq\rho_{R}$
we find final configurations $\rho_{f}$ with multiple minimizing solutions.
Fig. \ref{fig:two_paths_bdi} gives an example of such a $\rho_{f}$. The two
paths leading to this configuration are also shown in Fig. \ref{fig:kls_cusp}%
(a), where we plot the values of $\rho\left(  x=2/3,t\right)  $ against
$\rho\left(  x=1/3,t\right)  $ of the same histories. In Fig.
\ref{fig:kls_cusp}(a) we also plot the numerically obtained region in the
$\rho_{f}\left(  1/3\right)  $ and $\rho_{f}\left(  2/3\right)  $ plane for
which multiple histories are found, as well as the switching line. The jump in
the gradient $\delta\phi/\delta\rho_{f}$ is clear in Fig. \ref{fig:qs_cusp}%
(b), which depicts lines of equal $\phi$.

The LDSs in the two models have many features in common. The phase diagrams in
Fig. \ref{fig:qs_cusp}(c) and \ref{fig:kls_cusp}(a) have a similar structure,
with the singularities appearing for similar final profiles $\rho_{f}$. There
is one important difference: In contrast to the QS model, in the BDI model a
finite difference of the boundary conditions $\rho_{R}-\rho_{L}$ is needed in
order for an LDS to exist. In both models, generally we find (data not shown)
that as the value $\rho_{R}-\rho_{L}$ is decreased, the region with multiple
solutions is pushed away from $\bar{\rho}$. However, in contrast to the QS
model where $\rho$ is unbounded, in the BDI model $\rho$ is bounded
($0\leq\rho\leq1$). Hence below some threshold value, no LDS is found in the
BDI model. Similarly, by tuning $\varepsilon$ and $\delta$ in the BDI model,
$D$ and $\sigma$ can be continuously varied from the simple symmetric
exclusion model with $D=1$ and $\sigma=2\rho\left(  1-\rho\right)  $, for
which the LDF $\phi$ is known to be smooth, to the model discussed above. The
singularity appears when the dip in $\sigma\left(  \rho\right)  /D\left(
\rho\right)  $ is deep enough (data not shown).

To summarize, in both models we find LDSs when the function $\sigma\left(
\rho\right)  /D\left(  \rho\right)  $ has a (deep enough) local minimum.
Numerical experiments indicate that this is indeed, more generally, the
requirement. Recall that by fluctuation-dissipation, the ratio is related to
the compressibility $\sigma\left(  \rho\right)  /D\left(  \rho\right)
=2k_{B}T\rho^{2}\kappa\left(  \rho\right)  $. The profiles where the LDS is
found always have a shape similar to that in Fig. \ref{fig:qs_cusp}(b) and
Fig. \ref{fig:two_paths_bdi}. Intuitively, the existence of multiple
locally-minimizing histories leading to the same $\rho_{f}$ is due to the
favorable action due to large $\sigma\left(  \rho\right)  $ on certain
trajectories, utilizing densities on either side of the minimum in
$\sigma\left(  \rho\right)  $. A similar argument can be given for the ratio
$\sigma\left(  \rho\right)  /D\left(  \rho\right)  $. The existence and exact
location of the LDS depends on the full functional form of $\sigma\left(
\rho\right)  $ and \thinspace$D\left(  \rho\right)  $. It would be of interest
to find precise criteria.

\section{Structure of cusp\label{sec:cusp_structure}}

As discussed above, the structure of the LDS is similar in different models.
Consider $\rho_{f}$ in some fixed 2d cross-section of the full phase-space,
e.g., the cross-section defined in Eq. (\ref{eq:sin_interp}). As can be seen
in Fig. \ref{fig:kls_cusp}(b), the switching line ends at a single profile (a
point in the cross-section), which we denote by $\rho_{f}^{cusp}\left(
x\right)  $. This is much like a first-order transition line ending at a
second order point. We now discuss the behavior of the LDF $\phi\left[
\rho_{f}\right]  $ near $\rho_{f}^{cusp}$, as a function of $\rho_{f}$ and
$N$. As we now show, in the simplest scenario $\phi$ behaves like in a Landau
mean-field second order phase transition, or a \textquotedblleft cusp
catastrophe\textquotedblright\ in the language of catastrophe theory
\cite{Berry,catastrophe_ref}.

The discussion builds on previous results pertaining to systems with few
degrees of freedom
\cite{Noise_non_lin_book,MaierSteinCuspBiforcates,Maier_Stein,Dykman_cusp_structure,Dykman}%
. The singularity structure is well understood in such systems, where
catastrophe theory is applicable. The extension to fields requires care, as we
show below, see discussion at the end of this section. We first present the
theoretical considerations. Appendices \ref{sec:QS_cusp_ctructure} and
\ref{sec:appendix_KLS} verify the prediction for the QS and BDI models.

\subsection{Multiple minima near the cusp}

The action $S\left[  \rho,J\right]  $ is a functional of both $\rho$ and $J$.
The dependence of $S$ on the current $J$ is quadratic, and at fixed $\rho$ the
minimum over $J$, subject to Eq. (\ref{eq:conservation}), is unique. It will
therefore be convenient to work with the action after $J$ has been minimized:%
\[
s\left[  \rho\right]  =\min_{J}S\left[  \rho,J\right]  \ .
\]%
%TCIMACRO{\FRAME{ftbpFU}{1.7918in}{1.1843in}{0pt}{\Qcb{Definition of quantities
%near the cusp. Dashed line - switching line.}}{\Qlb{fig:cusp_tip_definitions}%
%}{cusp_tip_definitions.eps}{\special{ language "Scientific Word";
%type "GRAPHIC";  maintain-aspect-ratio TRUE;  display "USEDEF";
%valid_file "F";  width 1.7918in;  height 1.1843in;  depth 0pt;
%original-width 1.4688in;  original-height 0.9641in;  cropleft "0";
%croptop "1";  cropright "1";  cropbottom "0";
%filename 'cusp_tip_definitions.eps';file-properties "XNPEU";}} }%
%BeginExpansion
\begin{figure}
[ptb]
\begin{center}
\includegraphics[
height=1.1843in,
width=1.7918in
]%
{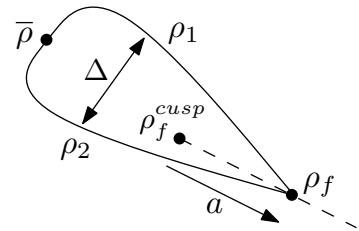}%
\caption{Definition of quantities near the cusp. Dashed line - switching
line.}%
\label{fig:cusp_tip_definitions}%
\end{center}
\end{figure}
%EndExpansion

For a given $\rho_{f}$ on the switching line there are two histories,
$\rho_{1}\left(  x,t\right)  $ and $\rho_{2}\left(  x,t\right)  $, which
minimize the action, as in Fig. \ref{fig:two_paths_bdi}. We introduce%
\[
a=\left[  \int\left(  \rho_{f}-\rho_{f}^{cusp}\right)  ^{2}dx\right]  ^{1/2}%
\]
as the distance of the final configuration $\rho_{f}$ from $\rho_{f}^{cusp}$,
see Fig. \ref{fig:cusp_tip_definitions}. We define a coordinate system,
$\left(  a,b\right)  $ on the cross-section, with $\rho_{f}^{cusp}$ at the
origin, $\hat{a}$ directed along the switching line and positive on the
switching line, and $\hat{b}$ orthogonal to the switching line. In analogy
with Landau mean-field theory, $a$ plays the role of $\left(  T_{c}-T\right)
$ and $b$ the role of the magnetic field.

Close to the cusp, when moving in the positive $a$ direction, for small enough
$\left\vert b\right\vert $,$\ s$\ has two locally-minimizing solutions,
$\rho_{1}$ and $\rho_{2}$. Let (see Fig. \ref{fig:cusp_tip_definitions})%
\begin{align}
\rho_{_{\left(  a,b\right)  }}^{avg}\left(  x,t\right)   &  =\frac{1}%
{2}\left[  \rho_{1}\left(  x,t\right)  +\rho_{2}\left(  x,t\right)  \right]
\ ,\nonumber\\
\delta\rho_{\left(  a,b\right)  }\left(  x,t\right)   &  =\frac{1}{2}\left[
\rho_{2}\left(  x,t\right)  -\rho_{1}\left(  x,t\right)  \right]
\ ,\nonumber\\
u_{\left(  a,b\right)  }\left(  x,t\right)   &  =\delta\rho_{\left(
a,b\right)  }/\left\Vert \delta\rho_{\left(  a,b\right)  }\right\Vert \ .
\label{eq:cusp_tip_defs}%
\end{align}
and
\[
\Delta=\left\Vert \delta\rho\right\Vert \ ,
\]
where we quantify the distance between two histories by $\left\Vert \delta
\rho\right\Vert ^{2}=\int\left[  \delta\rho\left(  x,t\right)  \right]
^{2}dxdt$. Here $\Delta$ plays the role of the amplitude of the order
parameter. Note that at the cusp $\Delta=0$.

On the switching line $b=0$ and $s\left[  \rho_{1}\right]  =s\left[  \rho
_{2}\right]  $ by definition. Hence%
\[
s_{\left(  a,b\right)  }\left(  q\right)  =s\left[  \rho_{\left(  a,b\right)
}^{avg}+qu_{\left(  a,b\right)  }\right]
\]
has two minima, at $q_{\min}=\pm\Delta$. $q$ is an \textquotedblleft order
parameter\textquotedblright\ interpolating between $\rho_{1}$ and $\rho_{2}$.
Close to $\rho_{f}^{cusp}$ the two minima should approach each other, and
merge to a single minimum at $\rho_{f}^{cusp}$. The simplest form which
captures this behavior and is analytical in $q,a$ and $b$ is%

\begin{equation}
\tilde{s}_{\left(  a,b\right)  }\left(  q\right)  =s_{\left(  a,b\right)
}\left(  q\right)  -s_{\left(  a,b\right)  }\left(  0\right)  =c_{4}%
q^{4}-ac_{2}q^{2}+c_{1}bq\ , \label{eq:S_of_xsi}%
\end{equation}
with $c_{1},c_{2},c_{4}>0$ constants. $\tilde{s}$ can also include higher
powers of $q$, which would not affect the behavior at small $a$. At small $a$
and $b=0$, $\tilde{s}\left(  q\right)  $ has two minima, at $q_{\min}%
\propto\pm\sqrt{a}$, hence $\Delta\propto\sqrt{a}$ in direct analogy with
Landau theory, with a mean-field exponent equal to 1/2.

In Appendix \ref{sec:appendix_KLS} we check the validity of
Eq.\ (\ref{eq:S_of_xsi}) on the BDI model. We show that it indeed holds, but
that higher order terms are still significant until close to the cusp point
($\left\Vert \rho_{f}-\rho_{f}^{cusp}\right\Vert \sim10^{-2}$). For the QS
model we use a different approach, see below.

\subsection{Effect of \textquotedblleft soft modes\textquotedblright}

When $\rho_{f}$ approaches $\rho_{f}^{cusp}$ from the positive $a$ direction,
the height of the action barrier separating the two local minima vanishes.
This means that the contribution of the paths close to the minimal paths is
enhanced. As we now show, this gives a new universal contribution to the
probability of $\rho_{f}^{cusp}$, scaling logarithmically in $N$%
\begin{equation}
P\left[  \rho_{f}^{cusp}\right]  \sim\exp\left(  -N\phi\left[  \rho_{f}%
^{cusp}\right]  +\frac{1}{4}\log N\right)  \ . \label{eq:rho_f_ln_correction}%
\end{equation}
The universal factor $1/4$ is known as the Berry index in catastrophe theory
\cite{Berry}.

To see this, we go back to the path integral formulation, $P\left[  \rho
_{f}\right]  =\int D\rho DJ\exp\left\{  -NS\left[  r,J\right]  \right\}  $. In
this expression, if the path integral is discretized then the measure is
$Dx=\prod_{i}\left(  \sqrt{N}dx_{i}\right)  $, where the $\sqrt{N}$ ensure
normalization. For given $\rho\left(  x,t\right)  $, $S$ is a quadratic
functional in $J$, so $J$ can be integrated out. For large $N$ a saddle-point
approximation gives%
\[
\ln\left(  \int DJ\exp\left\{  -NS\left[  \rho,J\right]  \right\}  \right)
=-N\min_{J}S\left[  \rho,J\right]  =-Ns\left[  \rho\right]  \ ,
\]
and the path integral now reads%
\begin{equation}
P\left[  \rho_{f}\right]  =\int D\rho\left(  x,t\right)  \exp\left\{
-Ns\left[  \rho\right]  \right\}  \ .
\end{equation}
Note that the correction to the saddle-point is $N^{0}$ in this case
\cite{foot_J_min}.

We focus on final configurations in the cross-section. Since close to the
cusp, when moving in the positive $a$ direction, $s\left[  \rho_{f}\right]
\,$\ has two solutions, this means that the Hessian matrix%
\begin{equation}
H\equiv\frac{\delta^{2}s}{\delta\rho_{\left(  x_{1},t_{1}\right)  }\delta
\rho_{\left(  x_{2},t_{2}\right)  }}\label{eq:Hessian_def}%
\end{equation}
has at least one vanishing eigenvalue for the optimal path ending at $\rho
_{f}^{cusp}$. (This is a standard result in catastrophe theory
\cite{catastrophe_ref}.) The corresponding eigenvector $u\left(  x,t\right)  $
is precisely $u_{\left(  a,b\right)  }$ defined in Eq. (\ref{eq:cusp_tip_defs}%
) for $\rho_{f}\rightarrow\rho_{f}^{cusp}$, i.e. $u\left(  x,t\right)
=u_{\left(  a\rightarrow0^{+},b=0\right)  }$. As $s$ is minimal, $H$ is
positive semi-definite, and its entire spectrum is non-negative. We now assume
that there is a single zero eigenvalue, followed by a finite gap. We then
split the histories as follows
\[
\rho\left(  x,t\right)  =\rho_{avg}\left(  x,t;a,b\right)  +qu\left(
x,t\right)  +\rho_{\perp}\left(  x,t;\rho_{f}\left(  a,b\right)  \right)  \ .
\]
Here $\rho_{avg}$ is defined as above when there are two minima, and is equal
to the minimal history when it is unique. $q$ is a numerical prefactor, and
all other contributions are included in $\rho_{\perp}\left(  x,t\right)  $.
Integrating out the $\rho_{\perp}$ directions we are left with an integral
over $q$%
\[
P\left[  \rho_{f}\left(  a,b\right)  \right]  \sim e^{-Ns_{\left(  a,b\right)
}\left(  0\right)  }N^{1/2}\int dq\exp\left\{  -N\tilde{s}_{\left(
a,b\right)  }\left(  q\right)  \right\}  \ .
\]
where $\tilde{s}_{\left(  a,b\right)  }\left(  q\right)  $ is defined in Eq.
(\ref{eq:S_of_xsi}). The $N^{1/2}$ comes from the definition of the path
integral measure. The form in Eq. (\ref{eq:S_of_xsi}) was argued on the basis
of the analyticity of $\tilde{s}\,$, justified by our assumption of the gap in
$H$. The integral
\[
\psi\left(  N,a,b\right)  =N^{1/2}\int dq\exp\left[  -N\left(  c_{1}%
bq-c_{2}aq^{2}+c_{4}q^{4}\right)  \right]
\]
is known as the \textquotedblleft cusp diffraction
catastrophe\textquotedblright\ \cite{Berry}. We note two of its properties:
(a) the \textquotedblleft metastability\textquotedblright\ region, where the
integrand has two local minima as a function of $q$ is bounded by $b\propto\pm
a^{2/3}$. (b) $\psi\left(  N,z_{1},z_{2}\right)  $\ has the scaling property
(with $u=N^{1/4}q$)
\begin{align}
\psi_{N}\left(  a,b\right)   &  =N^{1/2}\int dq\exp\left[  -N\left(
c_{1}bq-c_{2}aq^{2}+c_{4}q^{4}\right)  \right]  \nonumber\\
&  =N^{1/4}\Psi\left(  N^{1/2}a,N^{3/4}b\right)  \ ,\label{eq:berry_arnold}%
\end{align}
where $\Psi\left(  \alpha,\beta\right)  =\int dv\exp\left[  -\left(
c_{1}\beta v-c_{2}\alpha v^{2}+c_{4}v^{4}\right)  \right]  $ has no $N$
dependence. At $\left(  a,b\right)  =\left(  0,0\right)  $ this becomes
$\psi\left(  0,0\right)  =N^{1/4}\Psi\left(  0,0\right)  $.

Therefore, at $\rho_{f}^{cusp}$ we have $\phi\left[  \rho_{f}^{cusp}\right]
=s_{\left(  a,b\right)  }\left(  0\right)  $, and $P\left[  \rho_{f}%
^{cusp}\right]  $ has an additional $N^{1/4}$ prefactor to the probability
distribution shown in Eq. (\ref{eq:rho_f_ln_correction}). This means that at
the cusp the exponentiated $N$ dependence has an additional non-analytic
contribution, scaling as $\log N$ with a universal prefactor. The exponents
$1/4,1/2,3/4$ in Eq. (\ref{eq:berry_arnold}) were introduced in \cite{arnold}
and \cite{Berry}. This implies that the $\log N$ correction in Eq.
(\ref{eq:rho_f_ln_correction}) affects the probability in an elongated region
of dimensions $\Delta a\times\Delta b\sim N^{-1/2}\times N^{-3/4}$ around
$\rho_{f}^{cusp}$.

The above analysis relies on the assumption that the Hessian spectrum has a
single zero mode followed by a gap. This can be generalized to situations
where there is a finite number of zero modes followed by a gap, using tools
from catastrophe theory. In such cases, more complicated singular structures
will appear at $\rho_{f}^{cusp}$, with modified universal exponents. The
existence of a gap is expected to always hold in systems with
finite-dimensional phase spaces. However, in the case of fields, where the
phase space is infinite dimensional, the Hessian may be gapless. Then the
analyticity of the action might fail altogether, as indeed happens in
equilibrium critical phenomena \cite{schulman}.

In Appendices \ref{sec:QS_cusp_ctructure} and \ref{sec:appendix_KLS} we show
that the assumptions indeed hold. Specifically, for the QS model, we prove
that for specific types of profiles the assumption of analyticity is
justified. In addition, we calculate the Hessian spectrum numerically, and
find a gap above a single zero mode. For the BDI model, we show numerically
that the action indeed has a Landau mean-field form, Eq.\ (\ref{eq:S_of_xsi}).

\section{The connection with finite-dimensional phenomena -- a toy
model\label{sec:toy_model}}

In an attempt to better understand the LDS in this system, we note that the
minimizing histories $\rho\left(  x,t\right)  $ in Fig.
\ref{fig:two_paths_bdi} appear to be quite smooth in $x$; this is sensible, as
the field is constantly diffusing, making enduring, sharp spatial gradients
improbable. This was studied in \cite{our_numerics}. The smoothness motivates
us to introduce toy models with a finite number of degrees of freedom, which
capture many of the essential features of the field models described above.

To this end we discretize the field $\rho\,\left(  x,t\right)  $, replacing it
with a vector $\rho_{i}\left(  t\right)  $, $i=1,..,N_{p}$, corresponding to
the density at the points $x_{i}$. Substituting Eq. (\ref{eq:J_def}) into Eq.
(\ref{eq:conservation}), the Langevin equation reads $\partial_{t}%
\rho=\mathbf{\nabla}\cdot\left[  D\left(  \rho\right)  \mathbf{\nabla}%
\rho+\sqrt{\sigma\left(  \rho\right)  }\eta\right]  $. For $x\in\left[
0,1\right]  $, we take $x_{i}=i/(N_{p}+1)$, and obtain $N_{p}$ coupled
Langevin equations
\begin{align}
\partial_{t}\rho_{i} &  =\left(  \Delta x\right)  ^{-2}\left[  D_{i+1/2}%
\left(  \rho_{i+1}-\rho_{i}\right)  -D_{i-1/2}\left(  \rho_{i}-\rho
_{i-1}\right)  \right]  \nonumber\\
&  +\left(  \Delta x\right)  ^{-1}\left[  \sqrt{\sigma_{i,i+1}}\eta
_{i+1/2}-\sqrt{\sigma_{i-1,i}}\eta_{i-1/2}\right]  \ ,
\end{align}
\newline where $\Delta x=(N_{p}+1)^{-1}$, and $\rho_{0},\rho_{N_{p}+1}$ are
assigned the boundary values $\rho_{B}\left(  0\right)  ,\rho_{B}\left(
1\right)  $ respectively, and $\left\langle \eta_{i+1/2}\eta_{j+1/2}%
\right\rangle =\left(  \Delta x\right)  N^{-1}\delta_{i,j}$. $D_{i,i+1}%
,\sigma_{i,i+1}$ are an appropriate choice for $D\left(  \rho\right)
,\sigma\left(  \rho\right)  $ for $x$ between $x_{i}$ and $x_{i+1}$. We choose
$D_{i,i+1}=\frac{1}{2}\left[  D\left(  \rho_{i}\right)  +D\left(  \rho
_{i+1}\right)  \right]  $. A similar choice can be made for $\sigma_{i,i+1}$,
but we use $\sigma_{i,i+1}=2D_{i,i+1}\left(  \rho_{i+1}-\rho_{i}\right)
\left[  f^{\prime}\left(  \rho_{i+1}\right)  -f^{\prime}\left(  \rho
_{i}\right)  \right]  ^{-1}$, where $f^{\prime}\left(  \rho\right)  $ is given
in Eq. (\ref{eq:free_energy_density}). This has the advantage that when the
BCs are equal, $\rho_{0}=\rho_{N_{p}+1}$, the system of Langevin equations
satisfies detailed-balance \cite{gardiner} with respect to the potential
$\phi\left(  \left\{  \rho_{i}\right\}  \right)  =\Delta x\sum_{i}f\left(
\rho_{i}\right)  $. The discrete Langevin equation converges at high $N_{p}$
to the field-theory (as can be seen by writing the action). As the histories
are smooth we expect rapid convergence for long wave lengths. We therefore use
the lowest non-trivial discretization, $N_{p}=2$, which can accommodate
non-equilibrium phenomena. In this case $\rho_{1,2}$ correspond to
\textquotedblleft coarse-grained\textquotedblright\ densities at $x=1/3,2/3$ respectively.

The minimizing histories for the toy model with $N_{p}=2$ can be obtained
using standard techniques from low-noise finite-dimensional systems. Using the
approach outlined in the introduction, a set of coupled ordinary differential
equation is obtained, whose solutions are the extremizing histories. The
equations are solved numerically using a shooting method\cite{Graham_Tel}.

In Fig. \ref{fig:kls_vs_toy} we check the above approach on the simple
symmetric exclusion model (SSEP) with $D=1$ and $\sigma=2\rho\left(
1-\rho\right)  $, which does not feature a cusp, and for the BDI model. We
plot trajectories of the toy version of the SSEP\ (Fig. \ref{fig:kls_vs_toy}%
(a)) and the BDI model (Fig. \ref{fig:kls_vs_toy}(b)) in the $\left(  \rho
_{1},\rho_{2}\right)  $ plane, against the $\left(  \rho_{1/3},\rho
_{2/3}\right)  $ trajectories of the full models. In addition, the
metastability region for the toy model and in the cross-section of the exact
dynamics are plotted. The qualitative picture is similar -- a metastability
region appears in the quadrant $\rho_{1/3}>0,\rho_{2/3}<0$, at approximately
the same location as in the exact field solution. This is expected to have a
close relation to the breaking of detailed-balance, which is easy to visualize
in the two-dimensional toy model. Define the two-variable Langevin equation
$dx_{i}/dt=K_{i}\left(  x_{1},x_{2}\right)  +\sum_{j=1,2}B_{ij}\eta_{j}$, and
let $\mathbf{Q}\equiv\mathbf{BB}^{T}$. Detailed-balance is satisfied if
$\mathbf{\nabla}\phi=\mathbf{Q}^{-1}\mathbf{K}$, or $\nabla\times\left(
\mathbf{Q}^{-1}\mathbf{K}\right)  =0$. Therefore, in two (phase-space)
dimensions $\omega=\nabla\times\left(  \mathbf{Q}^{-1}\mathbf{K}\right)  $ is
a scalar which quantifies the breaking of detailed balance. It is shown in
Fig. \ref{fig:kls_toy_db_break}. Minimizing trajectories passing through
regions with $\omega>0$ ($\omega<0$) bend counter-clockwise (clockwise).
Therefore, a gradient in $\omega$ can cause trajectories to focus and cross,
creating a cusp. A similar picture has been discussed for other
finite-dimensional systems \cite{Dykman}.

In summary, the above analysis shows that much of the phenomena observed can
be captured by simplified finite-dimensional models, making concrete
connections to previous works on such models.%

%TCIMACRO{\FRAME{ftbpFU}{3.6521in}{1.5452in}{0pt}{\Qcb{(a) Toy model paths
%(thin lines) compared with cross-sections of exact paths for the SSEP model.
%(b) Same as (a) for BDI model. (c) Cusp area for toy vs. exact in the BDI
%model. BCs are $\rho_{L}=0.2,\rho_{R}=0.9$.}}{\Qlb{fig:kls_vs_toy}%
%}{kls_vs_toy.eps}{\special{ language "Scientific Word";  type "GRAPHIC";
%maintain-aspect-ratio TRUE;  display "USEDEF";  valid_file "F";
%width 3.6521in;  height 1.5452in;  depth 0pt;  original-width 8.7148in;
%original-height 3.6618in;  cropleft "0";  croptop "1";  cropright "1";
%cropbottom "0";  filename 'KLS_vs_toy.eps';file-properties "XNPEU";}} }%
%BeginExpansion
\begin{figure}
[ptb]
\begin{center}
\includegraphics[
height=1.5452in,
width=3.6521in
]%
{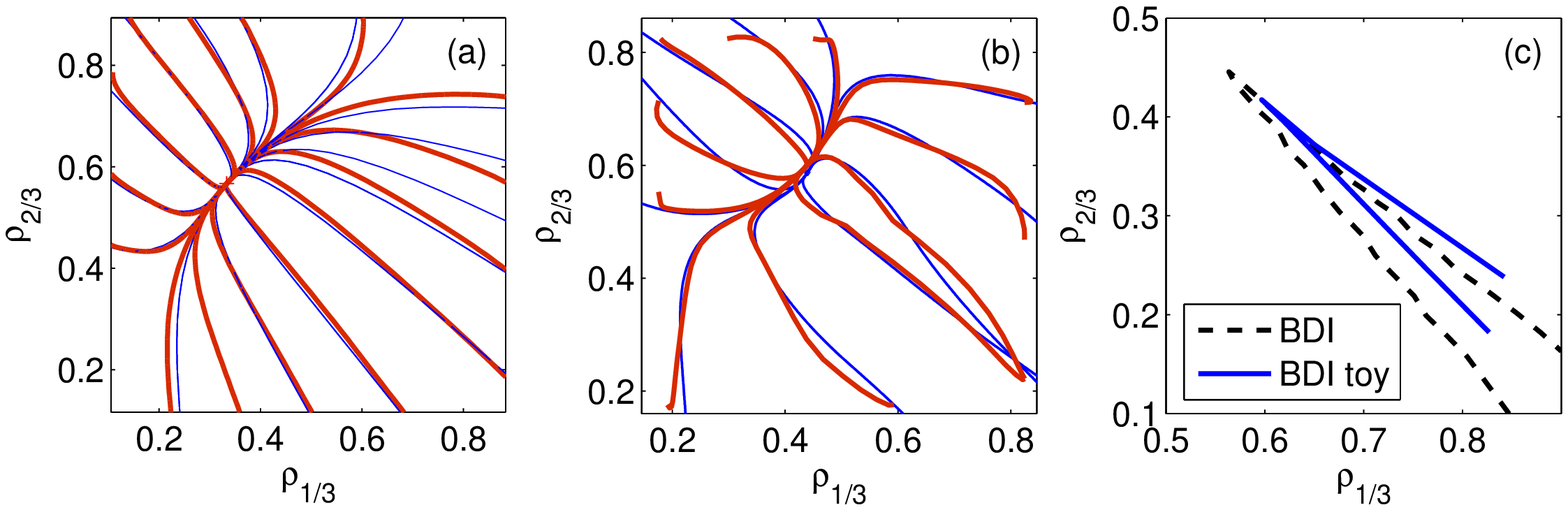}%
\caption{(a) Toy model paths (thin lines) compared with cross-sections of
exact paths for the SSEP model. (b) Same as (a) for BDI model. (c) Cusp area
for toy vs. exact in the BDI model. BCs are $\rho_{L}=0.2,\rho_{R}=0.9$.}%
\label{fig:kls_vs_toy}%
\end{center}
\end{figure}
%EndExpansion
%TCIMACRO{\FRAME{ftbpFU}{2.55in}{2.0239in}{0pt}{\Qcb{Map of $\omega$, the
%measure for breaking of detailed-balance, in the BDI model. Clockwise currents
%for positive $\omega$. In black: selected baths.}}{\Qlb{fig:kls_toy_db_break}%
%}{kls_toy_db_break.eps}{\special{ language "Scientific Word";
%type "GRAPHIC";  maintain-aspect-ratio TRUE;  display "USEDEF";
%valid_file "F";  width 2.55in;  height 2.0239in;  depth 0pt;
%original-width 8.7148in;  original-height 3.6618in;  cropleft "0";
%croptop "1";  cropright "1";  cropbottom "0";
%filename 'KLS_toy_DB_break.eps';file-properties "XNPEU";}} }%
%BeginExpansion
\begin{figure}
[ptb]
\begin{center}
\includegraphics[
height=2.0239in,
width=2.55in
]%
{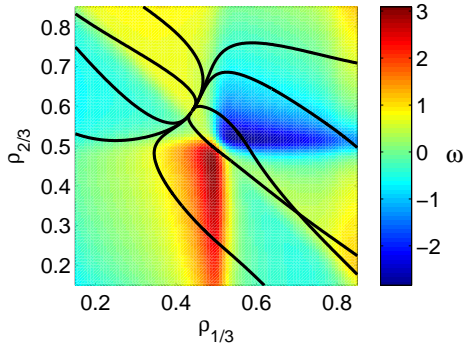}%
\caption{Map of $\omega$, the measure for breaking of detailed-balance, in the
BDI model. Clockwise currents for positive $\omega$. In black: selected
baths.}%
\label{fig:kls_toy_db_break}%
\end{center}
\end{figure}
%EndExpansion

\section{Discussion\label{sec:discussion}}

Many interesting questions remain to be studied. First, it would be of
interest to find precise conditions or bounds for the occurrence of the LDSs,
depending on the boundary conditions and model parameters. In addition, a
simple picture for the mechanism leading to the existence of multiple
histories ending at the same final profile is lacking. The mechanism suggested
by the current authors in \cite{ours_short} was flawed \cite{foot_sleeve_flaw}.

The singularity described above has the simplest possible structure. More
complicated singularities are in principle possible. For instance, in
catastrophe theory, richer structures have been analyzed. Observing them in
full requires one to look at higher dimensional cross-sections. It would be
interesting to find which of them exist in diffusive models, and for which
models. Even more exciting possibilities exist: Fields allow for the
possibility that the Hessian discussed in Sec. \ref{sec:cusp_structure} is
gapless. Thus, it may even be possible to find singularities that are beyond
the realm of catastrophe theory.

Finally, it would be very interesting to look at these LDSs in higher
dimensional systems. This is now possible using the numerical technique
described in \cite{our_numerics}, and used in the present paper.

\bigskip

\emph{Acknowledgments} - We are grateful for discussions with B. Derrida, J.
Kurchan, O. Raz and J. Tailleur. This research was funded by the BSF, ISF, and
IRG grants.

\bigskip%

%TCIMACRO{\TeXButton{appendix}{\appendix}}%
%BeginExpansion
\appendix
%EndExpansion

\section{Calculating $D\left(  \rho\right)  $ and $\sigma\left(  \rho\right)
$ for the driven Ising model\label{sec:appendix_KLS_D_sig}}

As shown in \cite{spohn_book,KLS_sigma}, for each parameter set $\left(
\varepsilon,\delta\right)  $ one can write implicit analytic equations for
$D\left(  \rho\right)  $ which can then be inverted numerically. Then
$\sigma\left(  \rho\right)  $ is obtained via the fluctuation-dissipation
relation, $\sigma\left(  \rho\right)  =2k_{B}T\rho^{2}\kappa\left(
\rho\right)  D\left(  \rho\right)  $ where $\kappa\left(  \rho\right)  $ is
the compressibility \cite{Derrida_review}. For equilibrium BCs this model
admits an Ising measure.

To find $D\left(  \rho\right)  $, we use the relation
\cite{foot_gradient_cond}
\[
D=\frac{1}{2\chi}\left(  \left\langle j_{i,i+1}\right\rangle +\left\langle
j_{i,i-1}\right\rangle \right)  =\frac{1}{\chi}\left\langle j_{i,i+1}%
\right\rangle
\]
where $\chi=\sum_{i}\left(  \left\langle n_{i}n_{0}\right\rangle -\rho
^{2}\right)  $ (related to the compressibility $\kappa\left(  \rho\right)  $
by $\chi=k_{B}T\rho^{2}\kappa\left(  \rho\right)  $), $j_{i,i+1}$ is the
current (number of particles per unit time) from site $i$ to site $i+1$. The
averages are taken with respect to the equilibrium probability distribution,
and $\left\langle j_{i,i+1}\right\rangle =\left\langle j_{i,i-1}\right\rangle
$ due to the symmetries in equilibrium. One then finds $\sigma$ using
\[
\sigma=2\chi D=2\left\langle j_{i,i+1}\right\rangle \ .
\]

To calculate $\left\langle j_{i,i+1}\right\rangle $ note that as the rates
depend on the four sites around a bond, we have that%
\begin{align*}
\left\langle j_{i,i+1}\right\rangle  &  =\left(  1+\delta\right)
P_{0100}+\left(  1+\varepsilon\right)  P_{1100}\\
&  +\left(  1-\varepsilon\right)  P_{0101}+\left(  1-\delta\right)  P_{1101}%
\end{align*}
where $P_{0100}$ is the probability of the pattern $0100$, and similarly for
others. Using the transfer-matrix technique \cite{KLS_sigma}, one can
calculate these probabilities and obtain%
\[
\left\langle j_{i,i+1}\right\rangle =\frac{\lambda\left[  1+\delta\left(
1-2\rho\right)  \right]  -\varepsilon\sqrt{4\rho\left(  1-\rho\right)  }%
}{\lambda^{3}}%
\]
where
\[
\lambda=\frac{1}{\sqrt{4\rho\left(  1-\rho\right)  }}+\left(  \frac{1}%
{4\rho\left(  1-\rho\right)  }-1+\frac{1-\varepsilon}{1+\varepsilon}\right)
^{1/2}\ .
\]
It remains to find $\rho$ and $\kappa$, which are both given in terms of
$\beta$ (recall that $\exp\left(  4\beta\right)  =\left(  1+\varepsilon
\right)  /\left(  1-\varepsilon\right)  $), and $h=\beta\mu$:%
\begin{align*}
\rho &  =\frac{1}{2}\left(  1+\frac{\sinh h}{\sqrt{e^{4\beta}+\sinh^{2}h}%
}\right)  ~,\\
\chi &  =\frac{e^{4\beta}\cosh h}{4\left(  e^{4\beta}+\sinh^{2}h\right)
^{3/2}}\ .
\end{align*}

In order to obtain $D\left(  \rho\right)  ,\sigma\left(  \rho\right)  $, we
calculate $\sigma\left(  h\right)  ,D\left(  h\right)  $ and $\rho\left(
h\right)  $ for a wide range of $h$, and numerically invert the last to find
$D\left(  \rho\right)  =D\left(  h\left(  \rho\right)  \right)  $ and
$\sigma\left(  \rho\right)  =\sigma\left(  h\left(  \rho\right)  \right)  $.
Fig. \ref{fig:models_sig_d}(a) shows $D\left(  \rho\right)  $ and
$\sigma\left(  \rho\right)  $ for $\left(  \varepsilon,\delta\right)  =\left(
0.05,0.995\right)  $.

As a check, we note that for the simple symmetric exclusion process
\cite{Derrida_review} $\delta=\varepsilon=0$, and one finds
\[
\rho=\frac{1}{2}\left(  1+\tanh h\right)  \ ;\ \ \chi=\frac{1}{4\cosh^{2}%
h}=\rho\left(  1-\rho\right)
\]
and
\[
\lambda=2\cosh h\ ;\ \ \ \left\langle j_{i,i+1}\right\rangle =\lambda
^{-2}=\frac{1}{4\cosh^{2}h}%
\]
so that $D=1$ and $\sigma=2\rho\left(  1-\rho\right)  $ \cite{Derrida_review}.

\section{Existence of multiple extremal solutions in the QS
model\label{sec:appendix_quad_sig}}

In this Appendix we prove that for the QS model, which has $D=1$ and
$\sigma\left(  \rho\right)  =\rho^{2}+1$, and for any non-equilibrium BCs,
there exists a LDS for some profiles. Here it will be far more convenient to
work in the domain $x\in\left[  -1,1\right]  $. The results in the new domain
are simply related to the results in the original domain
\cite{foot_QS_change_vars}.

The BCs to Eq. (\ref{eq:general_ODE}) are denoted by $\rho_{-1}\equiv\rho_{L}$
and $\rho_{+1}\equiv\rho_{R}$.

\begin{claim}
For any BCs $\rho_{-1}\neq\rho_{+1}$, there exists a profile $\rho_{f}\left(
x\right)  $ for which Eq. (\ref{eq:general_ODE}) is satisfied by more than one
solution with $g\left(  \pm1\right)  =\rho_{\pm1}$.
\end{claim}

\begin{proof}
Using the symmetries $\rho\rightarrow-\rho$ and $x\rightarrow-x$ it is enough
to consider the case $\rho_{-1}<\rho_{1}$, and $0<\rho_{1}$.

We proceed by an explicit construction of $\rho_{f}$. That is, given
$\rho_{\pm1}$ we construct a function $\rho_{f}\left(  x\right)  $ for which
Eq. (\ref{eq:general_ODE}) is satisfied by more than one function $g\left(
x\right)  $, which also satisfies the boundary-conditions. The profile
$\rho_{f}\left(  x\right)  $ will be a piecewise-constant function composed of
two flat regions, of the form%
\begin{equation}
\rho_{f}\left(  x\right)  =\left\{
\begin{array}
[c]{ccc}%
\rho_{A} &  & -1<x<0\\
\rho_{B} &  & 0<x<1
\end{array}
\right.  \ ,\label{eq:step_rho}%
\end{equation}
where $\rho_{A},\rho_{B}$ are (constant) numbers which specify $\rho
_{f}\left(  x\right)  $, see Fig. \ref{fig:many_sols}. Note that $\rho
_{f}\left(  x\right)  $ does not have to be continuous, nor to satisfy the
BCs, hence $\rho_{A},\rho_{B}$ are not restricted in any way. The solutions
$g\left(  x\right)  $ will be put together by solving Eq.
(\ref{eq:general_ODE}) for $x<0$ and $x>0$ separately (each with its
corresponding boundary condition), and matching the solutions by demanding
that $g\left(  x\right)  $ and $g^{\prime}\left(  x\right)  $ are continuous
at $x=0$.%
%TCIMACRO{\FRAME{ftbpFU}{2.7245in}{1.6698in}{0pt}{\Qcb{Density profile
%$\rho\left(  x\right)  $\ of the step form (Eq. (\ref{eq:step_rho})), and
%three $g\left(  x\right)  $ solutions. Here $\rho_{A}=4,\rho_{B}=-5,\rho
%_{-1}=-2,\rho_{1}=3$.}}{\Qlb{fig:many_sols}}{quad_sig_many_sols.eps}%
%{\special{ language "Scientific Word";  type "GRAPHIC";
%maintain-aspect-ratio TRUE;  display "USEDEF";  valid_file "F";
%width 2.7245in;  height 1.6698in;  depth 0pt;  original-width 10.4582in;
%original-height 6.3667in;  cropleft "0";  croptop "1";  cropright "1";
%cropbottom "0";  filename 'quad_sig_many_sols.eps';file-properties "XNPEU";}}
%}%
%BeginExpansion
\begin{figure}
[ptb]
\begin{center}
\includegraphics[
height=1.6698in,
width=2.7245in
]%
{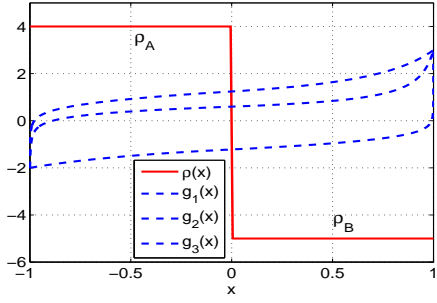}%
\caption{Density profile $\rho\left(  x\right)  $\ of the step form (Eq.
(\ref{eq:step_rho})), and three $g\left(  x\right)  $ solutions. Here
$\rho_{A}=4,\rho_{B}=-5,\rho_{-1}=-2,\rho_{1}=3$.}%
\label{fig:many_sols}%
\end{center}
\end{figure}
%EndExpansion

For a region with constant $\rho\left(  x\right)  =\bar{\rho}$, and given
$g\left(  x_{1}\right)  $, Eq. (\ref{eq:general_ODE}) has an (implicit)
analytic solution%
\begin{equation}
\int_{g\left(  x_{1}\right)  }^{g\left(  x\right)  }\frac{e^{\bar{\rho
}\operatorname{atan}\psi}}{\sqrt{1+\psi^{2}}}d\psi=c_{1}\left(  x-x_{1}%
\right)  \ ,\label{eq.analytic_sol}%
\end{equation}
where $c_{1}$ is a free constant. Differentiating both sides with respect to
$x$ we find%
\begin{equation}
g^{\prime}\left(  x\right)  \frac{e^{\bar{\rho}\operatorname{atan}g\left(
x\right)  }}{\sqrt{1+g\left(  x\right)  ^{2}}}=c_{1}%
\end{equation}
and using Eq. (\ref{eq.analytic_sol}) for $c_{1}$, $g^{\prime}\left(
x\right)  $ reads%
\begin{equation}
g^{\prime}\left(  x\right)  =\frac{\sqrt{1+g\left(  x\right)  ^{2}}}%
{e^{\bar{\rho}\operatorname{atan}\phi\left(  x\right)  }}\frac{1}{x-x_{1}}%
\int_{g\left(  x_{1}\right)  }^{g\left(  x\right)  }\frac{e^{\bar{\rho
}\operatorname{atan}\psi}}{\sqrt{1+\psi^{2}}}d\psi\ .
\end{equation}
Note that $c_{1}$ no longer appears in this equation. Instead, this is a
relation between $g\left(  x\right)  $ and $g^{\prime}\left(  x\right)  $. Let
$g_{A}\left(  x\right)  $ be the solution given in Eq. (\ref{eq.analytic_sol})
with $x_{1}=-1$, $g\left(  x_{1}\right)  =\rho_{-1}$ and $\bar{\rho}=\rho_{A}%
$:
\begin{equation}
\int_{\rho_{-1}}^{g_{A}\left(  x\right)  }\frac{e^{\rho_{A}\operatorname{atan}%
\psi}}{\sqrt{1+\psi^{2}}}d\psi=c_{A}\left(  x+1\right)  \ ,
\end{equation}
for $-1<x<0$. This defines a one-parameter family of solutions, according to
the value of $c_{A}$. Similarly, $g_{B}\left(  x\right)  $ is defined by
\begin{equation}
\int_{\rho_{1}}^{g_{B}\left(  x\right)  }\frac{e^{\rho_{B}\operatorname{atan}%
\psi}}{\sqrt{1+\psi^{2}}}d\psi=c_{B}\left(  x-1\right)
\end{equation}
for $0<x<1$. Any solution of Eq. (\ref{eq:general_ODE}) with $\rho\left(
x\right)  $ of the step form defined in Eq. (\ref{eq:step_rho}) is composed of
solutions $g_{A}\left(  x\right)  ,g_{B}\left(  x\right)  $ satisfying
$g_{A}\left(  0\right)  =g_{B}\left(  0\right)  $ and $g_{A}^{\prime}\left(
0\right)  =g_{B}^{\prime}\left(  0\right)  $. The derivatives at $x=0$ are
given by%
\begin{align}
g_{A}^{\prime}\left(  0\right)   &  =\frac{\sqrt{1+g_{A}\left(  0\right)
^{2}}}{e^{\rho_{A}\operatorname{atan}g_{A}\left(  0\right)  }}\int_{\rho_{-1}%
}^{g_{A}\left(  0\right)  }\frac{e^{\rho_{A}\operatorname{atan}\psi}}%
{\sqrt{1+\psi^{2}}}d\psi\ ,\nonumber\\
g_{B}^{\prime}\left(  0\right)   &  =\frac{\sqrt{1+g_{B}\left(  0\right)
^{2}}}{e^{\rho_{B}\operatorname{atan}g_{A}\left(  0\right)  }}\int
_{g_{B}\left(  0\right)  }^{\rho_{1}}\frac{e^{\rho_{B}\operatorname{atan}\psi
}}{\sqrt{1+\psi^{2}}}d\psi\ .\label{eq:dphiAB_expressions}%
\end{align}
We note that:

\begin{enumerate}
\item[(a)] $g^{\prime}\left(  x\right)  $ does not change sign. As we are
interested in solutions with $\rho_{-1}<\rho_{1}$, we only need to consider
solutions with $g^{\prime}\left(  x\right)  \geq0$.

\item[(b)] From (a) it follows that $\rho_{-1}\leq g_{A}\left(  0\right)
=g_{B}\left(  0\right)  \leq\rho_{1}$.

\item[(c)] It also follows that if $g_{A}\left(  0\right)  =\rho_{-1}$ then
$g_{A}^{\prime}\left(  0\right)  =0$, and if $g_{B}\left(  0\right)  =\rho
_{1}$ then $g_{B}^{\prime}\left(  0\right)  =0$. Similarly, if $g_{A}\left(
0\right)  =\rho_{1}$ then $g_{A}^{\prime}\left(  0\right)  >0$, and if
$g_{B}\left(  0\right)  =\rho_{-1}$ then $g_{B}^{\prime}\left(  0\right)  >0$.
\end{enumerate}

Consider now $g_{A}^{\prime}\left(  0\right)  $ and $g_{B}^{\prime}\left(
0\right)  $ as a function of $g\left(  0\right)  $. A solution $g\left(
x\right)  ~$on the entire segment $\left[  -1,1\right]  $ is obtained when
$g_{A}^{\prime}\left(  0\right)  =g_{B}^{\prime}\left(  0\right)  $ for the
same $g\left(  0\right)  $. Remark (c)\ ensures that they cross at least once;
But they may also cross more than once, see Fig. \ref{fig:phi_crossings}. The
number of crossings depends on $\rho_{A},\rho_{B}$. We will show that there
always exist $\rho_{A},\rho_{B}$ for which the graphs cross more than once.%
%TCIMACRO{\FRAME{ftbpFU}{3.0659in}{2.0976in}{0pt}{\Qcb{$g_{A}^{\prime}\left(
%0\right)  $ and $g_{B}^{\prime}\left(  0\right)  $ plotted as functions of
%$g\left(  0\right)  $ for two functions $\rho\left(  x\right)  $. In the upper
%pane the graphs cross only once, indicating a single $g\left(  x\right)
%$-solution. In the lower pane, done with parameters of Fig.
%(\ref{fig:many_sols}), they cross three times, resulting in three different
%$g$-solutions. (Note that the $g$-values at these crossings indeed correspond
%to $g\left(  0\right)  $ of the solutions in Fig. (\ref{fig:many_sols})). BCs
%for both panels are $\rho_{-1}=-3,\rho_{+1}=5$. Upper panel: $\rho_{A}%
%=3,\rho_{B}$ $=-2$, lower panel: $\rho_{A}=4,\rho_{B}$ $=-5$.}}%
%{\Qlb{fig:phi_crossings}}{quad_sig_phi_crossings.eps}%
%{\special{ language "Scientific Word";  type "GRAPHIC";
%maintain-aspect-ratio TRUE;  display "USEDEF";  valid_file "F";
%width 3.0659in;  height 2.0976in;  depth 0pt;  original-width 10.4582in;
%original-height 6.3667in;  cropleft "0";  croptop "1";  cropright "1";
%cropbottom "0";
%filename 'quad_sig_phi_crossings.eps';file-properties "XNPEU";}} }%
%BeginExpansion
\begin{figure}
[ptb]
\begin{center}
\includegraphics[
height=2.0976in,
width=3.0659in
]%
{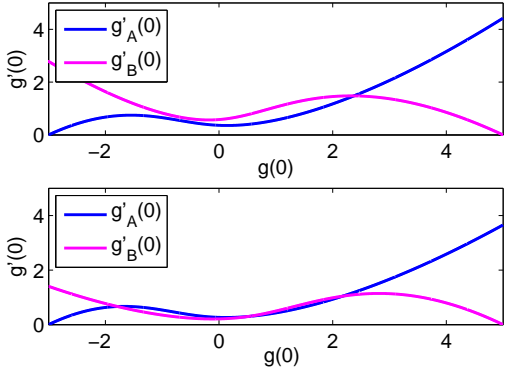}%
\caption{$g_{A}^{\prime}\left(  0\right)  $ and $g_{B}^{\prime}\left(
0\right)  $ plotted as functions of $g\left(  0\right)  $ for two functions
$\rho\left(  x\right)  $. In the upper pane the graphs cross only once,
indicating a single $g\left(  x\right)  $-solution. In the lower pane, done
with parameters of Fig. (\ref{fig:many_sols}), they cross three times,
resulting in three different $g$-solutions. (Note that the $g$-values at these
crossings indeed correspond to $g\left(  0\right)  $ of the solutions in Fig.
(\ref{fig:many_sols})). BCs for both panels are $\rho_{-1}=-3,\rho_{+1}=5$.
Upper panel: $\rho_{A}=3,\rho_{B}$ $=-2$, lower panel: $\rho_{A}=4,\rho_{B}$
$=-5$.}%
\label{fig:phi_crossings}%
\end{center}
\end{figure}
%EndExpansion

Motivated by the fact that the cusp singularities always appear at the lower
right corner of our phase-space cross-sections, we consider the limit where
$\rho_{A}$ is a large positive number, and $\rho_{B}\simeq-\rho_{A}$.

Denote by $g_{A}^{\prime}\left[  g\left(  0\right)  ;\rho_{A};\rho
_{-1}\right]  $ the value of $g_{A}^{\prime}\left(  0\right)  $ as a function
of $g\left(  0\right)  ,\rho_{A}$ and $\rho_{-1}$ , and similarly
$g_{B}^{\prime}\left[  g\left(  0\right)  ;\rho_{B};\rho_{1}\right]  $. The
analysis which follows is done for $g_{A}^{\prime}\left(  0\right)  $; similar
results are obtained for $g_{B}^{\prime}\left(  0\right)  $ since
$g_{B}^{\prime}\left[  g\left(  0\right)  ;\rho_{B};\rho_{1}\right]
=g_{A}^{\prime}\left[  -g\left(  0\right)  ;-\rho_{B};-\rho_{1}\right]  $. To
better understand $g_{A}^{\prime}\left[  g\left(  0\right)  ;\rho_{A}%
;\rho_{-1}\right]  $ at large $\rho_{A}$, we plot $g_{A}^{\prime}\left[
g\left(  0\right)  \right]  $ for $\rho_{A}=100$ and different $\rho_{-1}$
values, see Fig. \ref{fig:dphi_large_rho}. As can be seen, the different
graphs rise quickly from $g_{A}^{\prime}\left(  0\right)  =0$ at $g\left(
0\right)  =\rho_{-1}$, and join a common function.%
%TCIMACRO{\FRAME{ftbpFU}{2.6398in}{1.6147in}{0pt}{\Qcb{$g_{A}^{\prime}\left[
%g\left(  0\right)  ;\rho_{A}=100;\rho_{-1}\right]  $ for different $\rho_{-1}%
%$. Gray line: large $\rho$-expansion, Eq. (\ref{eq:dphiA_large_rho}).}%
%}{\Qlb{fig:dphi_large_rho}}{quad_sig_dphi_large_rho.eps}%
%{\special{ language "Scientific Word";  type "GRAPHIC";
%maintain-aspect-ratio TRUE;  display "USEDEF";  valid_file "F";
%width 2.6398in;  height 1.6147in;  depth 0pt;  original-width 10.4582in;
%original-height 6.3667in;  cropleft "0";  croptop "1";  cropright "1";
%cropbottom "0";
%filename 'quad_sig_dphi_large_rho.eps';file-properties "XNPEU";}} }%
%BeginExpansion
\begin{figure}
[ptb]
\begin{center}
\includegraphics[
height=1.6147in,
width=2.6398in
]%
{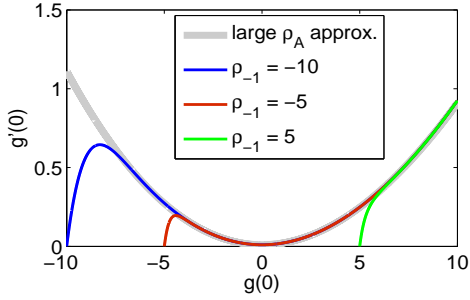}%
\caption{$g_{A}^{\prime}\left[  g\left(  0\right)  ;\rho_{A}=100;\rho
_{-1}\right]  $ for different $\rho_{-1}$. Gray line: large $\rho$-expansion,
Eq. (\ref{eq:dphiA_large_rho}).}%
\label{fig:dphi_large_rho}%
\end{center}
\end{figure}
%EndExpansion
This is formulated by the following Lemma:

\begin{lemma}
Expanding around $\rho_{A}\rightarrow\infty$, we have for any $g\left(
0\right)  >\rho_{-1}$
\begin{equation}
g_{A}^{\prime}\left[  g;\rho_{A};\rho_{-1}\right]  =\frac{1+g^{2}}{\rho_{A}%
}-\frac{g\left(  1+g^{2}\right)  }{\rho_{A}^{2}}+O\left(  \rho_{A}%
^{-3}\right)  \ \text{.}\label{eq:dphiA_large_rho}%
\end{equation}

\end{lemma}

\begin{proof}
Rewrite Eq. (\ref{eq:dphiAB_expressions}) for $g_{A}^{\prime}\left(  0\right)
$ as%
\[
g_{A}^{\prime}\left(  0\right)  =\sqrt{1+g^{2}}\int_{\rho_{-1}}^{g}%
\frac{e^{-\rho_{A}\left[  \operatorname{atan}g-\operatorname{atan}\psi\right]
}}{\sqrt{1+\psi^{2}}}d\psi\ ,
\]
where here and in the next equation $g$ stands for $g_{A}\left(  0\right)  $.
For $\rho_{A}\rightarrow\infty$ a saddle-point approximation can be preformed.
The exponent $-\rho_{A}\left[  \operatorname{atan}\phi-\operatorname{atan}%
\psi\right]  $ is dominated by small values of $g-\psi$, i.e. close to the
upper bound of the integral, and $\operatorname{atan}\phi-\operatorname{atan}%
\psi$ can be expanded to second order
\[
\operatorname{atan}g-\operatorname{atan}\psi=\frac{g-\psi}{1+g^{2}}%
+g\frac{\left(  g-\psi\right)  ^{2}}{\left(  1+g^{2}\right)  ^{2}}+O\left(
\left(  g-\psi\right)  ^{3}\right)  \ .
\]
In addition, the denominator $\left(  1+\psi^{2}\right)  ^{-1/2}$ is expanded
to second order around $g$. The resulting expression involves\ Gaussian
integrals which can be integrated, with the lower integration limit set to
$-\infty$. Finally, we expand the result (containing error-functions, etc.) to
second order in $1/\rho_{A}$ around $\rho_{A}\rightarrow\infty$, and obtain
Eq.\ (\ref{eq:dphiA_large_rho}).
\end{proof}

The expression in Eq. (\ref{eq:dphiA_large_rho}) does not depend on $\rho
_{-1}$, as expected from the reasoning alluding to Fig.
\ref{fig:dphi_large_rho}. Similarly, for $g_{B}^{\prime}\left[  g\left(
0\right)  ;\rho_{B};\rho_{1}\right]  $ we have, for $\rho_{B}\rightarrow
-\infty$,%
\begin{equation}
g_{B}^{\prime}\left[  g;\rho_{B};\rho_{1}\right]  =-\frac{1+g^{2}}{\rho_{B}%
}+\frac{g\left(  1+g^{2}\right)  }{\rho_{B}^{2}}+O\left(  \rho_{B}%
^{-3}\right)  \ \text{.}\label{eq:dphiB_large_rho}%
\end{equation}

We are now in a position to construct $\rho_{f}\left(  x\right)  $ with three
solution to $g$: given the BCs, choose some crossing value $g_{c}\in\left(
\rho_{-1},\rho_{1}\right)  $ (this is where the non-equilibrium condition
$\rho_{-1}\neq\rho_{1}$ enters). For a given $\rho_{A}$ the condition
$g_{A}^{\prime}\left[  g_{c};\rho_{A};\rho_{-1}\right]  =g_{B}^{\prime}\left[
g_{c};\rho_{B};\rho_{1}\right]  $, together with Eqs.
(\ref{eq:dphiA_large_rho}),(\ref{eq:dphiB_large_rho}) reads, to second order
in $\rho_{A}^{-1},\rho_{B}^{-1}$,%
\begin{equation}
\frac{1+g_{c}^{2}}{\rho_{A}}-\frac{g_{c}\left(  1+g_{c}^{2}\right)  }{\rho
_{A}^{2}}=-\frac{1+g_{c}^{2}}{\rho_{B}}+\frac{g_{c}\left(  1+g_{c}^{2}\right)
}{\rho_{B}^{2}}\ ,
\end{equation}
or%
\begin{equation}
\frac{1}{\rho_{A}}\left(  1-\frac{g_{c}}{\rho_{A}}\right)  =-\frac{1}{\rho
_{B}}\left(  1-\frac{g_{c}}{\rho_{B}}\right)  \ .
\end{equation}
Solving this for $\rho_{B}$ we find%
\begin{equation}
\rho_{B}=\frac{-1-\sqrt{1+4\frac{g_{c}}{\rho_{A}}\left(  1-\frac{g_{c}}%
{\rho_{A}}\right)  }}{2\left(  1-\frac{g_{c}}{\rho_{A}}\right)  }%
\ .\label{eq.rho_B_root}%
\end{equation}
\newline%
%TCIMACRO{\FRAME{ftbpFU}{3.0286in}{2.5169in}{0pt}{\Qcb{Constructing a solution.
%Choose a crossing value $g_{c}$ (here $g_{c}=-1$) and some large $\rho_{A}$.
%$\rho_{B}$ is fixed so that the crossing is approximately at $g_{c}$ (Upper
%pane). The full solutions will feature this crossing with two more crossings,
%close to the boundaries.}}{\Qlb{fig:design_cross}}{quad_sig_design_cross.eps}%
%{\special{ language "Scientific Word";  type "GRAPHIC";
%maintain-aspect-ratio TRUE;  display "USEDEF";  valid_file "F";
%width 3.0286in;  height 2.5169in;  depth 0pt;  original-width 10.4582in;
%original-height 6.3667in;  cropleft "0";  croptop "1";  cropright "1";
%cropbottom "0";
%filename 'quad_sig_design_cross.eps';file-properties "XNPEU";}} }%
%BeginExpansion
\begin{figure}
[ptb]
\begin{center}
\includegraphics[
height=2.5169in,
width=3.0286in
]%
{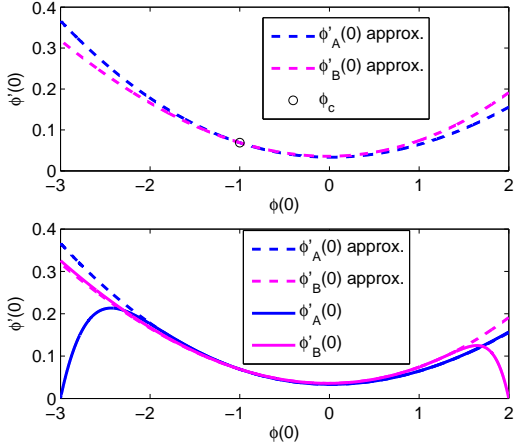}%
\caption{Constructing a solution. Choose a crossing value $g_{c}$ (here
$g_{c}=-1$) and some large $\rho_{A}$. $\rho_{B}$ is fixed so that the
crossing is approximately at $g_{c}$ (Upper pane). The full solutions will
feature this crossing with two more crossings, close to the boundaries.}%
\label{fig:design_cross}%
\end{center}
\end{figure}
%EndExpansion
Fixing $g_{c}$, the expansions Eqs. (\ref{eq:dphiA_large_rho}%
),(\ref{eq:dphiB_large_rho}) guarantee that as $\rho_{A}$ grows, with
$\rho_{B}$ given by Eq. (\ref{eq.rho_B_root}), a crossing point will appear in
a neighborhood of $g_{c}$, and approach $g_{c}$ for $\rho_{A}\rightarrow
\infty$. The root of the quadratic equation leading to Eq.
(\ref{eq.rho_B_root}) was chosen so that $g_{A}^{\prime},g_{B}^{\prime}$ will
cross from $g_{A}^{\prime}\left[  g\right]  >g_{B}^{\prime}\left[  g\right]  $
for $g<g_{c}$, to $g_{A}^{\prime}\left[  g\right]  <g_{B}^{\prime}\left[
g\right]  $ for $g>g_{c}$, see Fig. \ref{fig:design_cross}. This, together
with $g_{A}^{\prime}\left[  \rho_{-1};\rho_{A};\rho_{-1}\right]  =0$ and
$g_{B}^{\prime}\left[  \rho_{1};\rho_{B};\rho_{1}\right]  =0$, guarantees that
the functions $g_{A}^{\prime}\left[  g;\rho_{A};\rho_{-1}\right]  $ and
$g_{B}^{\prime}\left[  g;\rho_{B};\rho_{1}\right]  $, will have two crossing
points in addition to the crossing point near $g=g_{c}$. For large\ $\rho_{A}$
these will be at values of $g$ close to $\rho_{-1}$ and $\rho_{+1}$, see Fig.
(\ref{fig:design_cross}).
\end{proof}

\section{Cusp Structure and Hessian spectrum in the QS
model\label{sec:QS_cusp_ctructure}}

In this appendix we study the cusp structure, and the spectrum of the Hessian
matrix $H$\ at $\rho_{f}^{cusp}$, for the QS model. First, we prove that for
the QS model with profiles defined as in Appendix \ref{sec:appendix_quad_sig},
the framework of catastrophe theory is applicable. More precisely, there
exists an analytical function $F\left(  \rho_{A},\rho_{B}\right)  $\ on the
two-dimensional cross-section parametrized by $\left(  \rho_{A},\rho
_{B}\right)  $ as in Appendix \ref{sec:appendix_quad_sig}, and for which every
extremum of $F$ corresponds to a single extremal history leading to $\rho
_{f}\left(  \rho_{A},\rho_{B}\right)  $, as defined in Eq. (\ref{eq:step_rho}).

To construct the function $F$, we note that Eqs. (\ref{eq:dphiAB_expressions})
give us analytical expressions $g_{A}^{\prime}\left(  0\right)  =f_{A}\left(
g_{0},\rho_{A},\rho_{B}\right)  $ and $g_{B}^{\prime}\left(  0\right)
=f_{B}\left(  g_{0},\rho_{A},\rho_{B}\right)  $. We drop the $\rho_{-1}%
,\rho_{1}$ dependence, which are kept fixed. Let
\[
F\left(  g_{0};\rho_{A},\rho_{B}\right)  \equiv\left[  g_{A}^{\prime}\left(
0\right)  -g_{B}^{\prime}\left(  0\right)  \right]  ^{2}\ .
\]
$F\left(  \phi_{0};\rho_{A},\rho_{B}\right)  $ is analytic in all its
variables, and $F\left(  g_{0};\rho_{A},\rho_{B}\right)  =0$ iff the solution
$g$ is an extremal solution. Moreover, in the vicinity of the cusp, $\partial
F/\partial g_{0}=0$ iff $F\left(  g_{0};\rho_{A},\rho_{B}\right)  =0$, i.e. is
a local minimum as a function of $g_{0}$. Therefore $F$ acts as a
\textquotedblleft gradient map\textquotedblright\ (in the sense of Catastrophe
Theory), with $\rho_{A},\rho_{B}$ the control variables, and $g_{0}$ the state
variable. Accordingly, the cusp structure (regions in $\left(  \rho_{A}%
,\rho_{B}\right)  $ where $F\left(  g_{0}\right)  =0$ has two solutions) is
expected to be mean-field.

To show how $F$ is used, we briefly review the argument for the cusp
structure, which is essentially a Landau mean-field argument. The cusp point
is a special point $\left(  \rho_{A}^{cusp},\rho_{B}^{cusp}\right)  $ where at
the minimal $g_{0}$, $\partial^{2}F/\partial g_{0}^{2}=0=\partial
^{3}F/\partial g_{0}^{3}$ hold (the two conditions explain why it is a point,
or a set of isolated points, in the $\left(  \rho_{A},\rho_{B}\right)  $
plane). Let $x_{A,B}=\rho_{A,B}-\rho_{A}^{cusp}$. In the vicinity of the cusp,
expand $F$ to forth-order (\emph{here the analyticity is crucial}), $F=\alpha
g_{0}^{4}+\beta g_{0}^{3}+\gamma g_{0}^{2}+\delta$, where $\alpha,\beta
,\gamma,\delta$ depend on $x_{A,B}$. We assume that $\alpha\neq0$; vanishing
$\alpha$ would be non-generic, i.e., could be remedied by a small change in
any additional parameters, such as the BCs $\rho_{-1},\rho_{1}$ or the noise
function $\sigma$. Then a local change of variables can be performed to bring
$F$ to the form $F=\frac{1}{4}g_{0}^{4}+ag_{0}^{2}+b$, where at the cusp
$a=0=b$. The region where $F$ has two local minima is bounded by $b=\pm
\frac{2}{3^{3/2}}a^{3/2}$. Near the cusp $a,b$ can be expanded to first order
in $x_{A,B}$, so the power-law relation between $b\propto\pm a^{3/2}$ will
apply to a rotated frame of $x_{A,B}$.

\subsection{Spectrum\label{sec:QS_spectra}}

As discussed in the main text, for the \textquotedblleft Landau
mean-field\textquotedblright, catastrophe theory to hold, one must have a gap
in the spectrum of the Hessian $H$. This can be tested numerically, by
evaluating the action $S$ for the extremal solution $\rho\left(  x,t\right)  $
leading to $\rho_{f}\left(  x\right)  $, and calculating the Hessian, Eq.
(\ref{eq:Hessian_def}), by varying $\rho$ jointly at pairs of points $\left(
x_{1},t_{1}\right)  \ $and $\left(  x_{2},t_{2}\right)  $. The eigenvalues of
$H$ can then be calculated for different $\rho_{f}$ profiles close to
$\rho_{f}=\rho_{cusp}$.

To calculate $H$ one thus needs to locate $\rho_{cusp}$. This is most easily
done for the QS model with BCs $\rho_{L}=-\rho_{R}$, where $\rho_{cusp}$ of
the form of Eq. (\ref{eq:sin_interp}) must have $\alpha_{1}=0$. Fig.
\ref{fig:quad_sigma_spectrum_3_m3} shows the bottom of the spectra of $H$ for
different $\rho_{f}$, starting from $\bar{\rho}$ and ending at $\rho_{cusp}$.
One can clearly see a single eigenvalue going to zero, in agreement with the
analysis in the paper. The rest of the eigenvalues remain away from zero,
without closing the gap. This validates the analysis carried out in the main
text for the QS model.%
%TCIMACRO{\FRAME{ftbpFU}{3.0193in}{2.7457in}{0pt}{\Qcb{Bottom of Hessian
%spectrum for the QS model at different $\rho_{f}$ profiles, equaly spaced
%between $\bar{\rho}$ (leftmost) to $\rho_{f}^{cusp}$ (rightmost). BCs are
%$\rho_{L}=-3,\rho_{R}=3$. A single eigenvalue approaches zero, while the gap
%above it is maintained.}}{\Qlb{fig:quad_sigma_spectrum_3_m3}}%
%{quad_sigma_spectrum_3_m3.eps}{\special{ language "Scientific Word";
%type "GRAPHIC";  maintain-aspect-ratio TRUE;  display "USEDEF";
%valid_file "F";  width 3.0193in;  height 2.7457in;  depth 0pt;
%original-width 5.1169in;  original-height 4.6501in;  cropleft "0";
%croptop "1";  cropright "1";  cropbottom "0";
%filename 'quad_sigma_spectrum_3_m3.eps';file-properties "XNPEU";}} }%
%BeginExpansion
\begin{figure}
[ptb]
\begin{center}
\includegraphics[
height=2.7457in,
width=3.0193in
]%
{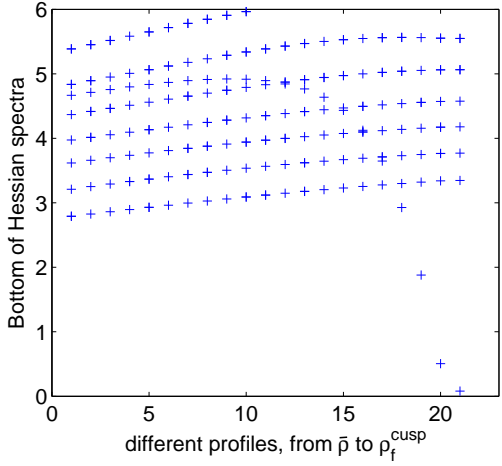}%
\caption{Bottom of Hessian spectrum for the QS model at different $\rho_{f}$
profiles, equaly spaced between $\bar{\rho}$ (leftmost) to $\rho_{f}^{cusp}$
(rightmost). BCs are $\rho_{L}=-3,\rho_{R}=3$. A single eigenvalue approaches
zero, while the gap above it is maintained.}%
\label{fig:quad_sigma_spectrum_3_m3}%
\end{center}
\end{figure}
%EndExpansion

\section{Cusp structure in the BDI model\label{sec:appendix_KLS}}

In this Appendix we check the validity of Eq. (\ref{eq:S_of_xsi}), which
predicts the structure of the cusp. In the BDI model the diagonalization of
the Hessian $H$ gave inconclusive results. We suspect that this is due to the
difficulty of locating $\rho_{cusp}$ with high precision in this model. As we
show now, the predictions of Sec. \ref{sec:cusp_structure} hold only very
close to $\rho_{f}^{cusp}$, when $\left\Vert \rho_{f}-\rho_{f}^{cusp}%
\right\Vert <10^{-2}$.

To compare Eq.\ (\ref{eq:S_of_xsi}) with numerics, it is more convenient to
use a different form, which does not require knowing the precise position of
$\rho_{f}^{cusp}$. Noting that $s_{\left(  a,0\right)  }\left(  \Delta\right)
=s\left[  \rho_{1}\right]  $ and $s_{\left(  a,0\right)  }\left(
-\Delta\right)  =s\left[  \rho_{2}\right]  $, we find that $\Delta\propto
\sqrt{a}$. Therefore we expect that for $b=0$%

\begin{equation}
s_{q}\left(  \Delta\right)  -s_{q=0}\left(  \Delta\right)  \propto\Delta
^{4}\left(  \frac{1}{4}y^{4}+\frac{1}{2}y^{2}\right)  \ , \label{eq:s_of_y}%
\end{equation}
where $y=q/\Delta$.

As an example we consider the boundary-driven Ising model, with $\left(
\varepsilon,\delta\right)  =\left(  0.05,0.995\right)  $ and $\rho
_{L}=0.2,\rho_{R}=0.8$. Examples of pairs of locally minimizing histories
leading to configurations on the switching line are shown in Fig.
\ref{fig:on_the_switch}.%
%TCIMACRO{\FRAME{ftbpFU}{2.5779in}{1.9502in}{0pt}{\Qcb{Pairs of locally
%minimizing histories, leading to points on the switching line (solid and
%dashed lines). The history leading to $\rho_{f}^{cusp}$ (circle)\ is also
%plotted (bold line).}}{\Qlb{fig:on_the_switch}}{on_the_switch.eps}%
%{\special{ language "Scientific Word";  type "GRAPHIC";
%maintain-aspect-ratio TRUE;  display "USEDEF";  valid_file "F";
%width 2.5779in;  height 1.9502in;  depth 0pt;  original-width 4.5014in;
%original-height 3.7963in;  cropleft "0";  croptop "1";  cropright "1";
%cropbottom "0";  filename 'on_the_switch.eps';file-properties "XNPEU";}} }%
%BeginExpansion
\begin{figure}
[ptb]
\begin{center}
\includegraphics[
height=1.9502in,
width=2.5779in
]%
{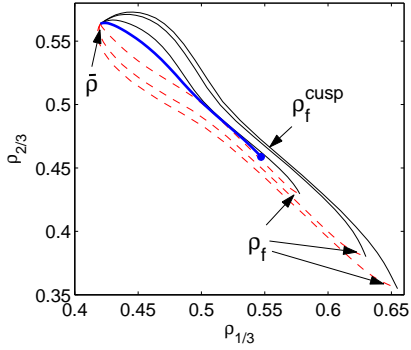}%
\caption{Pairs of locally minimizing histories, leading to points on the
switching line (solid and dashed lines). The history leading to $\rho
_{f}^{cusp}$ (circle)\ is also plotted (bold line).}%
\label{fig:on_the_switch}%
\end{center}
\end{figure}
%EndExpansion

Fig. \ref{fig:quartic_polynomials_fit}(a) shows the function $\left[
s_{q}\left(  \Delta\right)  -s_{q=0}\left(  \Delta\right)  \right]
/\Delta^{4}$ as a function of $y$ at $a\sim10^{-2}$, together with a quartic
fit, which shows clear deviations from this form. This means that even at this
distance $\rho_{f}-\rho_{f}^{cusp}\sim10^{-2}$ there are significant
contributions of higher powers to $S_{q}$. Due to these higher powers plotting
$\left[  s_{q}\left(  \Delta\right)  -s_{q=0}\left(  \Delta\right)  \right]
/\Delta^{4}$ vs. $y$ for different $a$ values in the range $5\cdot10^{-3}\leq
a\leq0.2$ does not collapse the data as expected. We therefore fit the
functions $s_{q}\left(  \Delta\right)  -s_{q=0}\left(  \Delta\right)  $ to
polynomials of order four and higher, and plot $c_{4}$, the prefactor of
$q^{4}$, as a function of $\Delta^{2}$, see Fig.
\ref{fig:quartic_polynomials_fit}(b). The expected power-law is not obtained
for a quartic fit, but improves when the fits include higher order terms, see
Fig. \ref{fig:c4_fits_both}(a). This means that $c_{4}$ in Eq.
(\ref{eq:S_of_xsi}) can indeed be taken to be constant. Finally, one can also
fit $s_{q}\left(  \Delta\right)  -s_{q=0}\left(  \Delta\right)  \propto
a^{2}\left(  \frac{1}{4}y^{4}+\frac{c}{2}y^{2}\right)  $, by fitting the
position of $\rho_{f}^{cusp}$, see Fig. \ref{fig:c4_fits_both}(b).%
%TCIMACRO{\FRAME{ftbpFU}{3.509in}{1.4673in}{0pt}{\Qcb{(a) The function
%$S_{q}\left(  \Delta\right)  -S_{q=0}\left(  \Delta\right)  /\Delta^{4}$ \ as
%a function of $y$ at $a=0.012$ (solid line). A clear deviation is seen from a
%fit to a quartic function (dashed line). (b) The function $S_{q}\left(
%\Delta\right)  -S_{q=0}\left(  \Delta\right)  /\Delta^{4}$ for different
%values of $a$ (dashed lines). Fitting the functions to a polynomial of power
%8, and plotting only the quartic part, the collapse improves significantly
%(solid lines).}}{\Qlb{fig:quartic_polynomials_fit}}%
%{quartic_polynomials_fit.eps}{\special{ language "Scientific Word";
%type "GRAPHIC";  maintain-aspect-ratio TRUE;  display "USEDEF";
%valid_file "F";  width 3.509in;  height 1.4673in;  depth 0pt;
%original-width 9.1332in;  original-height 3.4911in;  cropleft "0";
%croptop "1.002532";  cropright "1";  cropbottom "0.002532";
%filename 'quartic_polynomials_fit.eps';file-properties "XNPEU";}} }%
%BeginExpansion
\begin{figure}
[ptb]
\begin{center}
\includegraphics[
trim=0.000000in 0.008839in 0.000000in -0.008839in,
height=1.4673in,
width=3.509in
]%
{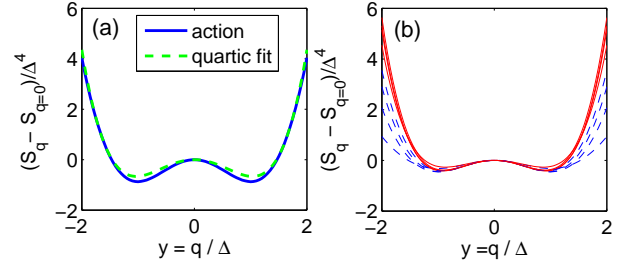}%
\caption{(a) The function $S_{q}\left(  \Delta\right)  -S_{q=0}\left(
\Delta\right)  /\Delta^{4}$ \ as a function of $y$ at $a=0.012$ (solid line).
A clear deviation is seen from a fit to a quartic function (dashed line). (b)
The function $S_{q}\left(  \Delta\right)  -S_{q=0}\left(  \Delta\right)
/\Delta^{4}$ for different values of $a$ (dashed lines). Fitting the functions
to a polynomial of power 8, and plotting only the quartic part, the collapse
improves significantly (solid lines).}%
\label{fig:quartic_polynomials_fit}%
\end{center}
\end{figure}
%EndExpansion
%TCIMACRO{\FRAME{ftbpFU}{3.686in}{1.6969in}{0pt}{\Qcb{Fits of $c_{4}$. (a)
%$c_{4}$ as a function of $\Delta^{2}$. (a) $c_{4}$ as a function of $a$.
%Guidelines (dashed lines) represent the expected slope of the functions.
%Extracting $c_{4}$ from fits that also include higher powers one obtains data
%that better fits the expected slope.}}{\Qlb{fig:c4_fits_both}}%
%{c4_fits_both.eps}{\special{ language "Scientific Word";  type "GRAPHIC";
%maintain-aspect-ratio TRUE;  display "USEDEF";  valid_file "F";
%width 3.686in;  height 1.6969in;  depth 0pt;  original-width 9.1332in;
%original-height 4.1602in;  cropleft "0";  croptop "1";  cropright "1";
%cropbottom "0";  filename 'c4_fits_both.eps';file-properties "XNPEU";}} }%
%BeginExpansion
\begin{figure}
[ptb]
\begin{center}
\includegraphics[
height=1.6969in,
width=3.686in
]%
{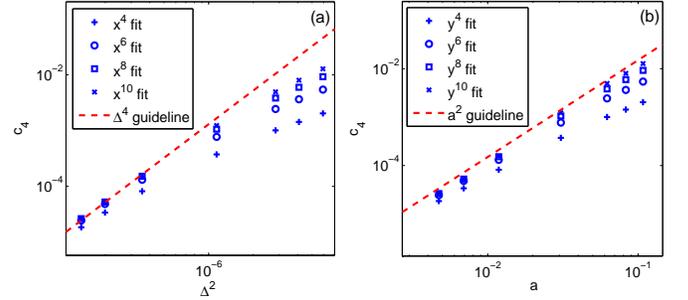}%
\caption{Fits of $c_{4}$. (a) $c_{4}$ as a function of $\Delta^{2}$. (a)
$c_{4}$ as a function of $a$. Guidelines (dashed lines) represent the expected
slope of the functions. Extracting $c_{4}$ from fits that also include higher
powers one obtains data that better fits the expected slope.}%
\label{fig:c4_fits_both}%
\end{center}
\end{figure}
%EndExpansion

The presence of strong higher powers as close as $a\sim10^{-2}$, see e.g. Fig.
\ref{fig:c4_fits_both}(b), can be understood as follows: the non-linear terms
come from the different action at the two paths $\rho_{1}$ and $\rho_{2}$. As
the distance $\Delta$ between the, scales as $\Delta\propto\sqrt{a}$,
$\rho_{1}-\rho_{2}$ at some space time point, can be of order $10^{-1}$ even
for $a\sim10^{-2}$, and the two minimizing paths can see very different
behavior of $D\left(  \rho\right)  ,\sigma\left(  \rho\right)  $ along the
paths. This sensitivity explains why computing $H$ directly is difficult: one
needs $\Delta$ to be small, but as $a\propto\Delta^{2}$, one needs the
distance $a$ from $\rho_{f}^{cusp}$ to be very small.

From the above we conclude that at the cusp there is a \emph{soft mode},\ a
direction along which the minimizing history $\rho^{cusp}\left(  x,t\right)  $
has a zero second derivative: $d^{2}s\left[  \rho^{cusp}\left(  x,t\right)
+qu_{a=0}\left(  x,t\right)  \right]  /dq^{2}=\left[  d^{2}s_{q}%
/dq^{2}\right]  _{a=0}=0$. This is indirect evidence that $H$ has at least one
vanishing eigenvalue.


\begin{thebibliography}{99}
\bibitem {foot_extremize_S}The paths which extremize \thinspace$S$ are given
by $\partial_{t}\rho=\partial_{x}^{2}\rho+\partial_{x}\left(  \sigma
\partial_{x}\hat{\rho}\right)  $ with $\hat{\rho}=\nabla\phi_{ext}$.
\bibitem {foot_J_min}The action $S=\int dxdt\frac{\left(  J+D\nabla
\rho\right)  ^{2}}{2\sigma}$ is quadratic in $J$, therefore the minimum of $S$
over $J$ is unique. The spectrum of $\delta^{2}S/\delta J_{1}\delta J_{2}$ is
gapped if $\sigma\left(  \rho\right)  $ is bounded away from zero on the path
$\rho\left(  x,t\right)  $. This justifies the saddle-point approximation
(i.e., higher order corrections to this equation are $O\left(  N^{0}\right)  $).
\bibitem {foot_sleeve_flaw} Referring to \cite{ours_short}, the estimate for
the LDF of the history which satisfies the symmetry did not take into account
the possibility that all the mass transfer occurs at the middle point $x=1/2$.
\bibitem {foot_gradient_cond}This holds for models that satisfy the
\textquotedblleft Gradient Condition\textquotedblright, see \cite{spohn_book}.
\bibitem {foot_QS_change_vars}Changing variables $x^{\prime}=2x,t^{\prime
}=4t,\rho^{\prime}=\rho$ and $J^{\prime}=J/2$, the continuity equation is
unchanged and the transformed action Eq. (\ref{eq:action}) reads $S^{\prime
}\left[  \rho^{\prime},J^{\prime}\right]  =2S\left[  \rho,J\right]  \,$.
Clearly, the number of extrema of the action for a given $\rho_{f}$ is not
affected by this change.
\end{thebibliography}
\end{document}